\documentclass{article}

\usepackage{amsmath}
\usepackage{amsfonts}
\usepackage{amssymb}
\usepackage{amsthm}
\usepackage{diagbox}
\usepackage{graphicx}
\usepackage[hidelinks]{hyperref}
\usepackage{enumitem}
\usepackage{eucal}
\usepackage{algorithm}
\usepackage{algpseudocode}
\usepackage{bm}
\usepackage{bbm}
\usepackage{multirow}
\usepackage{mathtools}

\DeclareMathOperator*{\argmax}{argmax}

\usepackage{etoolbox}
\usepackage{tikz}
\usetikzlibrary{arrows}

\usepackage{lineno}

\newtheorem*{prop*}{Proposition}

\newcommand*\linenomathpatch[1]{%
  \cspreto{#1}{\linenomath}%
  \cspreto{#1*}{\linenomath}%
  \csappto{end#1}{\endlinenomath}%
  \csappto{end#1*}{\endlinenomath}%
}

\linenomathpatch{equation}
\linenomathpatch{gather}
\linenomathpatch{multline}
\linenomathpatch{align}
\linenomathpatch{alignat}
\linenomathpatch{flalign}

\newcommand{\bY}{\mathbf{Y}}

\newcommand{\pcsp}[2]{{#1}\hspace{-0.09em}\rightarrow\hspace{-0.09em}{#2}}

\newcommand{\sdag}{\mathcal{D}}

\newcommand{\sscl}[2]{(#1,\underline{#2})}

\newcommand{\subsplit}[2]{\{\{#1\},\{#2\}\}}

\newcommand{\iqtree}{\texttt{iqtree}}
\newcommand{\bito}{\texttt{bito}}

\newcommand{\raxml}{\texttt{RAxML}}
\newcommand{\mrbayes}{\texttt{MrBayes}}
\newcommand{\slurm}{\texttt{SLURM}}
\newcommand{\citet}[1]{\cite{#1}}
\newcommand{\citep}[1]{\cite{#1}}

\newlength{\arrowwidth}

\newcommand{\beginsupplement}{%
        \setcounter{table}{0}
        \renewcommand{\thetable}{S\arabic{table}}%
        \setcounter{figure}{0}
        \renewcommand{\thefigure}{S\arabic{figure}}%
     }

\title{Finding high posterior density phylogenies by systematically extending a directed acyclic graph}

\author{
Chris Jennings-Shaffer$^{1}$,
David H Rich$^{1}$, 
Matthew Macaulay$^{2}$, \\
Michael D Karcher$^{3}$, 
Tanvi Ganapathy$^{1}$,\\
Shosuke Kiami$^{1}$,
Anna Kooperberg$^{1}$,  
Cheng Zhang$^{4}$, 
Marc A Suchard$^{5,6,7}$,\\
Frederick A Matsen IV$^{1,8,9,*}$
}
\date{%
\small
$^{1}$Public Health Sciences Division, Fred Hutchinson Cancer Research Center, Seattle, Washington, USA;\\
$^{2}$University of Technology Sydney, Australian Institute for Microbiology \& Infection, Sydney, Australia;\\
$^{3}$Department of Math \& Computer Science, Muhlenberg College, Allentown, Pennsylvania, USA;\\
$^{4}$School of Mathematical Sciences, Peking University, Beijing, China;\\
$^{5}$Department of Human Genetics, University of California, Los Angeles, USA;\\
$^{6}$Department of Computational Medicine, University of California, Los Angeles, USA;\\
$^{7}$Department of Biostatistics, University of California, Los Angeles, USA;\\
$^{8}$Department of Genome Sciences, University of Washington, Seattle, USA;\\
$^{9}$Howard Hughes Medical Institute, Fred Hutchinson Cancer Research Center, Seattle, Washington, USA;\\
$^{*}$Corresponding author. Computational Biology Program, Fred Hutchinson Cancer Research Center, 1100 Fairview Ave. N., Mail stop: S2-140, Seattle, WA 98109-1024 \\
Email addresses: \texttt{cjennin2@fredhutch.org, david.rich27@gmail.com, matt.macaulay101@gmail.com, michaelkarcher@muhlenberg.edu, tganapat@caltech.edu, kiami.sho@gmail.com, alkooperberg@gmail.com, chengzhang@math.pku.edu.cn, msuchard@ucla.edu, matsen@fredhutch.org}
}

\begin{document}
\allowdisplaybreaks

\maketitle

\clearpage

\begin{abstract}
Bayesian phylogenetics typically estimates a posterior distribution, or aspects thereof, using Markov chain Monte Carlo methods.
These methods integrate over tree space by applying local rearrangements to move a tree through its space as a random walk.
Previous work explored the possibility of replacing this random walk with a systematic search, but was quickly overwhelmed by the large number of probable trees in the posterior distribution.
In this paper we develop methods to sidestep this problem using a recently introduced structure called the subsplit directed acyclic graph (sDAG).
This structure can represent many trees at once, and local rearrangements of trees translate to methods of enlarging the sDAG.
Here we propose two methods of introducing, ranking, and selecting local rearrangements on sDAGs to produce a collection of trees with high posterior density.
One of these methods successfully recovers the set of high posterior density trees across a range of data sets.
However, we find that a simpler strategy of aggregating trees into an sDAG in fact is computationally faster and returns a higher fraction of probable trees.
\end{abstract}

\section*{Introduction}

Despite decades of work, Bayesian phylogenetics remains a computationally challenging problem.
Existing methods are based on the Markov chain Monte Carlo (MCMC) algorithm.
These methods begin with a tree, which may be random or generated by another method (e.g., parsimony), and propose random modifications to the tree. 
A random modification is accepted with probability proportional to the Metropolis-Hastings ratio of the new tree and current tree, which biases the chain towards trees with higher likelihood and prior probability.
We use the term \emph{topology} for the graph theoretic tree structure and \emph{tree} for a topology with branch lengths.
Because the high-confidence region of topologies is a tiny subspace of a space of super-exponential size, most of these random modifications will result in a significantly worse tree, and thus substantial modifications are overwhelmingly discarded.
This leads to low acceptance rates, which may reduce efficiency. 
Thus, although MCMC is a robust and flexible algorithm, it is inherently limited in its ability to efficiently infer phylogenetic posterior distributions.

Maximum likelihood methods take a different approach, and primarily work to find the maximum likelihood tree systematically rather than randomly.
These methods are typically iterative and apply local rearrangements at each step to improve the likelihood of the current highest likelihood tree.
These methods are substantially faster than MCMC methods, but do not attempt to characterize the entire credible set of possible trees or topologies.

Is it possible to combine these two approaches, in which one systematically infers an approximate Bayesian posterior distribution on trees?
Although the phylogenetic likelihood can only be evaluated on a tree with branch lengths, one can define the likelihood of a topology to be the likelihood of the tree given by the topology along with optimal branch lengths.
Previous work \cite{phylotopo} performed a systematic and parallel exploration of tree space by local rearrangements of the visited topologies and collecting only the resulting topologies above some likelihood threshold.
Normalizing the likelihoods of all visited topologies gives an approximation of the Bayesian posterior distribution. 
Although this formed an interesting proof of concept, it was ultimately defeated by there being too many high quality trees to do likelihood-based branch length optimization on each one.

More recent work \cite{sbn, vbpi, vbpi2, GP, hdag, Berling2024} has developed computational structures that are capable of storing and manipulating many trees or topologies at once.
We use terminology from \cite{GP} and call this structure a \emph{subsplit Directed Acyclic Graph} or \emph{sDAG} for short.
Different bifurcating tree-graphs in the sDAG correspond to different tree topologies (Figure~\ref{fig:treecomparison}).

In this paper, we develop systematic search strategies with sDAGs instead of topologies, with the goal of finding the smallest sDAG that contains a credible set of topologies.
Because sDAGs can represent many topologies at once, we avoid the problem of having too many topologies to consider individually.
Our approach in this work is to extend the nearest neighbor interchange (NNI) operation on topologies to an operation on sDAGs.
As detailed in the Methods section, we develop NNIs as an operation to enlarge an sDAG rather than a move from topology to topology. 

This approach requires a means of deciding whether an NNI is worth applying to the sDAG\@. 
One cannot directly apply classical phylogenetic criteria that judge a single topology at a time, because a single NNI operation can add many topologies to the sDAG at once (see Methods section).
We apply two approaches.
Both approaches associate branch lengths with each edge of the sDAG, which means that there is a one-to-one correspondence between topologies in the sDAG and trees in the sDAG\@.

The first approach, \emph{top pruning}, implements the idea that one would like to apply NNIs that generate at least one good topology; this corresponds to the idea of collecting a credible set of topologies and merging them into a sDAG\@.
In slightly more detail, top pruning maintains \emph{choice maps} that can be used recursively to get a ``best'' tree containing the central edge of any given NNI\@, and branch lengths of new additions to the sDAG are optimized to maximize the likelihood of the ``best'' tree containing the central edge of that NNI\@.
Thus, the likelihood of NNIs in the top pruning case is the classical Felsenstein phylogenetic likelihood of that ``best'' tree associated to the NNI.
(Here ``best'' is put in quotes because the ``best'' tree is an approximation to the maximum likelihood tree.)

The second, \emph{generalized pruning}, or \emph{GP}, implements the idea that one would like to apply NNIs that generate many probable topologies, but with a composite-like approximation to the marginal likelihood.
One can use this marginal likelihood to optimize branch lengths for newly added sDAG edges, as well as to decide if an NNI is worth applying to the sDAG under the GP criterion.
Specifically, the likelihood for an NNI is this GP marginal likelihood with optimized branch lengths.
All of this is made possible by a recently-developed algorithm to calculate the marginal composite likelihood across topologies~\cite{GP} for which computation time scales linearly in the number of edges of the sDAG\@.

When applied to benchmark data sets, we find that top pruning performs significantly better than generalized pruning in terms of discovering a credible set of topologies.
However, neither method delivers a major advance in terms of finding a small subsplit DAG containing a credible set when compared to aggregating an sDAG from a short run of \mrbayes~\cite{Ronquist2012-hi}.
Although this aggregation approach was taken to generate sDAGs in past work (e.g., \cite{sbn, GP, vbpi}), here we show, using a variety of data sets, that this aggregation strategy gives good representation of the posterior distribution without being over-diffuse.

\section*{Methods}

We begin by introducing the sDAG\@. 
Although it was described briefly as part of previous work~\cite{GP}, which assumed that such a structure was given, here our goal is to infer such a structure, so we will spend more time on developing and motivating the idea.
Other recent independent work~\cite{Berling2024} has developed a related but different structure.

\subsection*{Introduction to the subsplit DAG}

\begin{figure}
	\centering
	\includegraphics[width=\textwidth]{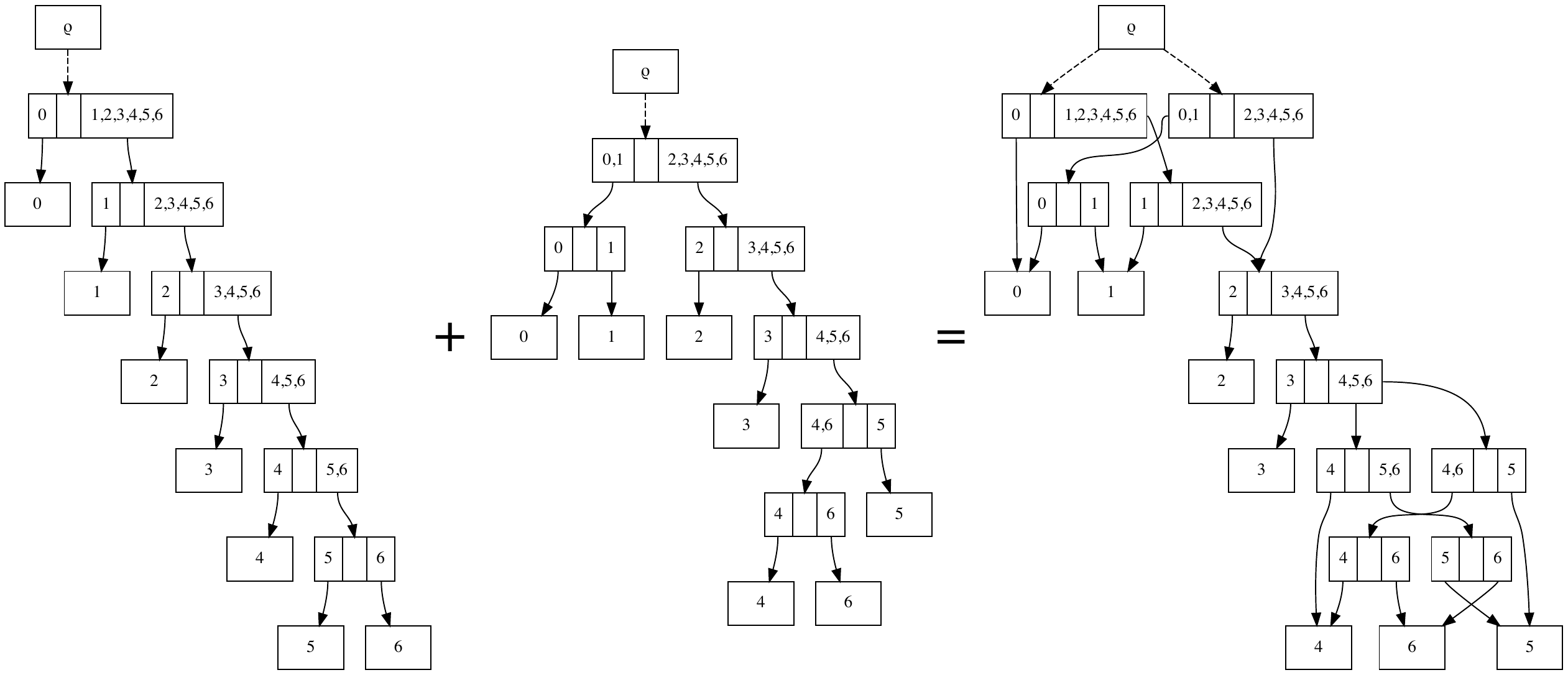}
	\caption{\
		The combination of two topologies $\tau_1$ and $\tau_2$ into a single sDAG\@.
		The sDAG (right) contains the union of the nodes and edges of the individual sDAGs from each of the topologies.
		It also contains additional topologies, such as the topology containing both $\subsplit{0}{1, 2, 3, 4, 5, 6}$ and $\subsplit{4, 6}{5}$, that are not present in the original set of two topologies.
	}
	\label{fig:treecomparison}
\end{figure}

We first describe how the subsplit DAG is a generalization of a single, rooted, bifurcating topology, and how the general case can be considered the union of a collection of topologies.
Take, for example, a single caterpillar topology on the taxon set $\{0, \ldots , 6\}$ (Figure~\ref{fig:treecomparison} left).

In this representation, we label each internal node with the two taxon sets that are leafward of each of the two edges coming from that internal node.
So, for example, the second node below the root $\rho$ has the taxon set $\{1\}$ on one side and the taxon set $\{2, 3, 4, 5, 6\}$ on the other, so the internal node is labeled with $\subsplit{1}{2, 3, 4, 5, 6}$.

We call such a bipartition of a subset of the taxon set a \emph{subsplit} (following~\citet{sbn, vbpi, vbpi2, GP}; subsplits were called \emph{partial splits} in~\citet{Redelings2012-ch}).
Each component of the subsplit is a \emph{clade} of the subsplit, and each of these clade components is referred to as a \emph{subsplit-clade}.
One can represent any rooted phylogenetic topology as an sDAG\@: each node is labeled with a subsplit describing the bipartition of the taxa leafward of that node, and each edge is directed in a leafward direction from the root.
More generally, an sDAG is a directed acyclic graph with subsplits as nodes and edges connecting parent subsplits to child subsplits partitioning individual clades in the parent subsplit.
To distinguish which subsplit-clade is partitioned along an edge, we write $\sscl{t}{X}\to s$, where $X$ is a subsplit-clade of the subsplit $t$, $s$ is a child subsplit of $t$, and $\bigcup(s)=X$.
We underline the subsplit-clade $X$ so that there is no ambiguity in which variables are subsplits and which are subsplit-clades.
For convenience, we assume there is an order on the taxa, which extends to an order on clades. 
We call the lesser subsplit-clade the left subsplit clade and the greater subsplit-clade the right subsplit-clade. 

We extend some common terminology for clades to sDAG edges. 
Suppose $t_1$, $t_2$, $s_1$, $s_2$ are subsplits, $e_1$ is the edge from $t_1$ to $s_1$, and $e_2$ is the edge from $t_2$ to $s_2$.
When $s_1=t_2$, we say $e_1$ is a parent edge of $e_2$ and say $e_2$ is a child edge of $e_1$; we may refine this by saying left child edge or right child edge, depending on which subsplit clade of $s_1$ is partitioned by $s_2$.
When $t_1=t_2$, but $s_1$ and $s_2$ partition distinct subsplit clades of $t_1$, we say $e_1$ and $e_2$ are sibling edges.

To build an sDAG encoding multiple topologies, we take the sDAG for each topology, then take the union of the nodes and edges in the two individual sDAG representations (Figure~\ref{fig:treecomparison}).
Thus an sDAG may contain a collection of topologies: any graph-theoretic-tree-structured subset of the nodes and edges in an sDAG that contains all of the leaves represents a tree topology\@.
Each subsplit-clade of outdegree greater than one requires a choice between the descending arrows.
For example, consider the edges leaving the subsplit-clade $\{4,5,6\}$ in Figure~\ref{fig:treecomparison}.
If we pick the edge leading to $\subsplit{4}{5,6}$, we will have a tree with $\{4\}$ branching off first, and if we pick the edge leading to $\subsplit{4,6}{5}$, we will have a tree with $\{5\}$ branching off first.

Note that if we build the sDAG from a collection of topologies, the sDAG may contain additional topologies beyond those used to build the sDAG (Figure~\ref{fig:treecomparison}).
In many respects this is a feature, not a bug: it allows us to expand the support of the sDAG combinatorially beyond the set of topologies used to build it.
On the other hand this can add topologies outside the credible set.
The balance between the advantage of additional topologies and disadvantage of ``false positive'' topologies will be a key consideration in our systematic search strategies.
However, as described next, we have finer control of the topology distributions than if we were to use the conditional clade distribution of~\citet{Hhna2012-pm} and~\citet{Larget2013-uo}.

\subsubsection*{Phylogenetic tree distributions}

A key application of the sDAG is to represent a probability distribution on phylogenetic topologies and trees.
If we assign such a probability distribution to the edges originating in each of the subsplit-clades, then we obtain a probability distribution on phylogenetic tree topologies.
For example, in Figure~\ref{fig:treecomparison} we can assign probabilities to the two options for the root split and to the two options for resolving the subsplit-clade $\{4,5,6\}$: $\subsplit{4}{5,6}$ and $\subsplit{4,6}{5}$.
Suppose the probability of $\subsplit{0,1}{2,3,4,5,6}$ is $0.4$ and the probability of $\subsplit{4}{5,6}$ is $0.3$, then the probability of the topology $((0,1),(2,(3,(4,(5,6)))))$ containing both of these subsplits is $0.4 \times 0.3 = 0.12$.
It is easy to see that assigning probability distributions to each set of edges leaving each subsplit-clade in the sDAG yields a (normalized) probability distribution on the topologies represented in the sDAG.

\begin{figure}
	\centering
	\includegraphics[width=0.4\textwidth]{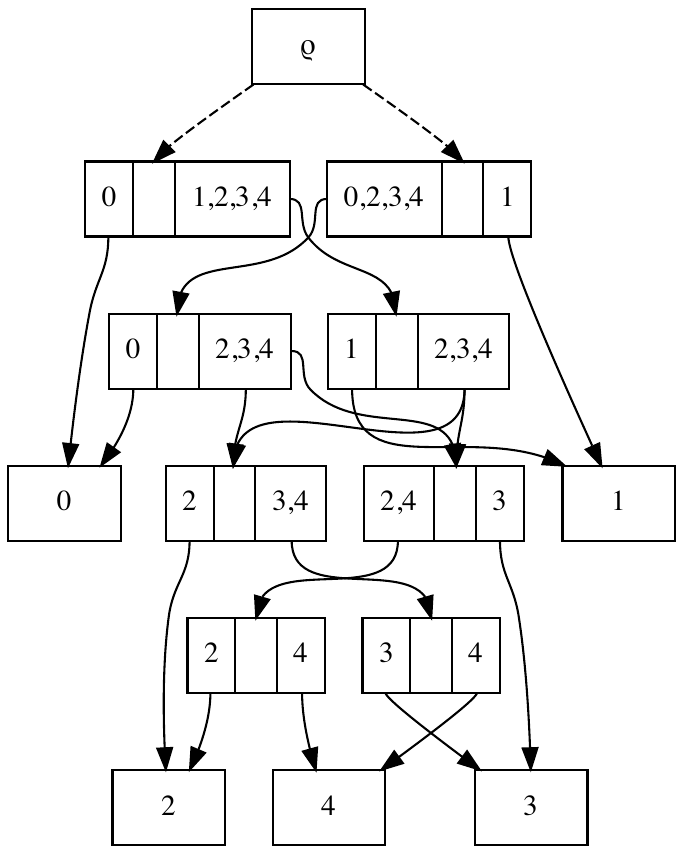}
	\caption{\
        An example showing how the sDAG is more flexible than the conditional clade distribution of~\citet{Hhna2012-pm,Larget2013-uo,Berling2024}.
        Specifically, we may have different splitting probabilities for the clade $\{2,3,4\}$ depending on which subsplit it is contained in (either $\subsplit{0}{2, 3, 4}$ or $\subsplit{1}{2, 3, 4}$), showing that our approach is a strict generalization of clade-conditional approaches.
	}
    \label{fig:sdagIsntCCD}
\end{figure}

Such a distribution is a generalization of previous work~\citep{Hhna2012-pm,Larget2013-uo} that led to topology distributions called \emph{conditional clade distributions} (CCDs); CCDs consist of a distribution of subsplits conditioned on a clade, i.e.\ a set of taxa.
Recently, these CCDs were attached to a directed acyclic graph structure similar to our sDAG \cite{Berling2024}.
The difference between this and the sDAG formulation is that the sDAG enables the expression of additional conditional dependencies on the sister clade (Figure~\ref{fig:sdagIsntCCD}).
We have shown such flexibility greatly improves fit~\citep{sbn}.
However, if those are not important, one can allow them to be independent of the sister clade, recovering conditional clade distributions in a graph structure.

In terms of probabilities of edges, the sDAG formulation takes a probability distribution on edges leaving each parent subsplit-clade, whereas the conditional clade distribution requires these distributions to be identical for parent subsplit-clades on the same clade.

If one has a sample of topologies, such as that from an MCMC algorithm, one can use it to fit the probabilities labeling the edges of the sDAG\@.
For an sDAG on rooted topologies, the probabilities for edges from a parent subsplit to its child subsplits are simply normalized frequency counts of the child subsplits in the sampled topologies that contain the parent subsplit.
If we desire a distribution on unrooted topologies, we can consider the sDAG containing all possible rootings of the topologies in the sample, and an expectation-maximization algorithm can be used to infer probabilities~\citep{sbn}.

To extend such a distribution on topologies to a distribution on trees, we attach parameterized distributions for branch lengths to the sDAG edges. 
Taking such an sDAG and inferring the branch length and subsplit distributions is a case of variational Bayesian phylogenetic inference. 
This approach was introduced and studied in \cite{vbpi,vbpi2}, taking the sDAG as a fundamental object, although it was described using different terminology: ``subsplit Bayesian networks.''
However, in this paper, as described below, we will be assigning a single fixed branch length to each edge of an sDAG, so that each topology in an sDAG corresponds to a phylogenetic tree.

\subsection*{Performing NNIs to the subsplit DAG}

\begin{figure}
	\centering
	\includegraphics[width=1\textwidth]{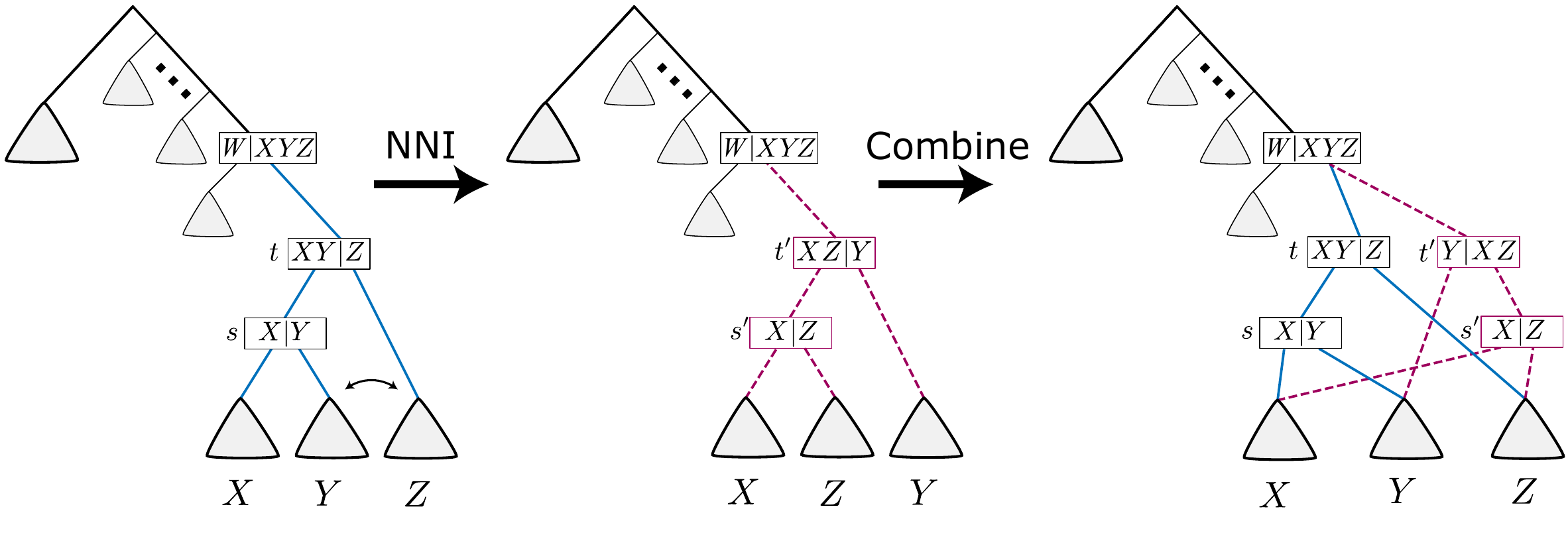}
	\caption{\
		Applying an NNI to a sDAG at subsplits $t$ and $s$, then combining the pre-NNI sDAG with the post-NNI sDAG into a single sDAG\@.
		Here $t'$ and $s'$ are the new nodes, and the red dashed edges are the new edges. Triangles are substructures in the sDAG, and $X$, $Y$, and $Z$ are clades.
	}
	\label{fig:dagNNI}
\end{figure}

The goal of this paper is to develop a systematic inference of the sDAG. 
In order to do so, we describe modifications of the sDAG that are analogous to the sort of modifications typically done to phylogenetic topologies.
\emph{Nearest-neighbor interchange} (NNI) is a common method to make minor topological modifications and grow the topological support.
An NNI swaps two subtopologies of adjacent subsplits to create a new topology (see left-hand arrow of Figure~\ref{fig:dagNNI}).

We can form an analogous operation on the sDAG, however, the NNI operation enlarges the sDAG rather than modifies it in place.
That is, on an sDAG, we can perform NNI on two subsplit-clades, then combine the pre-NNI sDAG with the post-NNI sDAG into a single new sDAG\@.
Consider Figure~\ref{fig:dagNNI}, where $t$ and $s$ are subsplits on the clade set $X\cup Y\cup Z$, with $t = \{X \cup Y,Z\}$ and $s = \{X,Y\}$.
Performing an NNI on clades $Y$ and $Z$ produces a new sDAG with subsplits $t'$ and $s'$, where $t' = \{X \cup Z,Y\}$ and $s'=\{X,Z\}$.

To construct the combined sDAG from the pre-NNI sDAG, we need only add subsplits $t'$, $s'$ (if not already in the sDAG) and edges (a) between the new nodes: $t' \to s'$, (b) parents of $t'$: $u \to t'$ for all $u$ where $u \to t$, and (c) descendants of the new subsplit-clades: $\sscl{s'}{X} \to u$ for all $u$ where $\sscl{s}{X} \to u$, $\sscl{t'}{Y} \to u$ for all $u$ where $\sscl{s}{Y} \to u$, and $\sscl{s'} {Z} \to u$ for all $u$ where $\sscl{t}{Z} \to u$.
The new edges described above will preserve all previously existing incoming and outgoing edges from each clade.
The addition of $t'$ and $s'$ creates a locally different splitting order between the $X$, $Y$, and $Z$ clades, leaving all other parts of the sDAG unmodified. 
We refer to the edge $t^\prime\to s^\prime$ as the \emph{central edge} of the NNI.

In our systematic inference methods, we maintain an sDAG with edges between all compatible subsplits, i.e.\ subsplits that can co-exist in a tree.
Such sDAGs exhibit favorable properties.
For example, when the pre-NNI sDAG has edges between all compatible subsplits, every new topology in the post-NNI sDAG is an NNI of a topology in the pre-NNI sDAG. 
Additionally, any new topology must contain the central edge of the NNI.
Both statements can fail when the original sDAG is missing compatible edges.
Details are given in the appendix.
However, after applying an NNI to an sDAG, the resulting sDAG may not have the maximum number of edges even when the original sDAG does. 
In particular, the sDAG may not have edges from $t^\prime$ to all compatible child subsplits and edges from all compatible parent subsplits to $s^\prime$.
Thus, when enlarging an sDAG via an NNI in our systematic search algorithms, we will include these additional edges.

\subsection*{Evaluating new additions to the sDAG}

Our inference algorithms proceed in a manner analogous to that for hill-climbing search of a single phylogenetic tree: evaluate all possible local modifications, and accept a modification according to an optimality criterion.
Thus, we require a means of evaluating an NNI of the sDAG to decide if we should apply it to the sDAG\@. 
Specifically, we describe two ways that generalize the calculation of a likelihood for a phylogenetic tree.
This is enabled by assigning a single fixed branch length to each edge of the sDAG, as described above, so that each topology in an sDAG corresponds to a phylogenetic tree.

Before describing these approaches in detail, we introduce some further notation.
We write $\bY$ for the given multiple sequence alignment written as a rectangular array, each sequence is one row.
The $i$th column of $\bY$ is $\bY^i$, which is the column vector containing the $i$th site of each sequence. 
As we deal only with fixed branch lengths, we write $p_\psi(\bY\mid \tau)$ for the phylogenetic likelihood for the data $\bY$, branch lengths $\psi$, and topology $\tau$: that is, this is the classical phylogenetic likelihood of a tree that has $\tau$ as the topology and branch lengths assigned according to $\psi$, assuming data $\bY$.
As commonly done, we assume site independence so that $p_\psi(\bY\mid \tau)=\prod_{i=1}^K p_\psi(\bY^i \mid \tau)$, where $K$ is the number of sites.
In this case, $p_\psi(\bY^i \mid \tau)$ is efficiently calculated by Felsenstein's pruning algorithm. 
In our implementation and benchmarking of our NNI-search algorithms, we use the Jukes-Cantor substitution model for simplicity.
We could use a general time reversible model, as long as we take fixed model parameters (equilibrium frequencies and substitution rates) for the sDAG.
One approach to infering these model parameters is to fit them on a single tree in the sDAG.
We also do not consider across-site rate variation, although this could be added to both of the algorithms here if desired.
Also, we optimize branch lengths rather than marginalize over them, in the interest of efficiency.
In Bayesian phylogenetics one typically considers a distribution of branch lengths, however this style of approximation has shown surprisingly good performance~\cite{Suchard2003-us,Anisimova2011-oh,Fourment2020-te}.

The two approaches, described in detail in the following two sections, are called ``top pruning'' and ``generalized pruning.''
Each approach is based on its own definition of likelihood, described roughly above and in more detail in the following sections.
However, for now, we note that generalized pruning is based on an across-tree marginalization so NNI operations are evaluated based on the sDAG as a whole, while top pruning is based on the (classical) likelihood of the single best additional tree enabled by the NNI.

\subsubsection*{Top Pruning}

The idea behind top pruning is to add the highest likelihood tree, obtained by an NNI, not previously in the sDAG.
Recall an NNI on an sDAG may introduce more than one new tree. 
Ideally one would select the NNI based on a computation like
\begin{gather}\label{eq:fullTopLikelihood}
\max_{\tau\in T_e} \prod_{j=1}^K p_\psi(\bY^j\mid \tau)
,
\end{gather}
where $T_e$ is the set of trees in the post-NNI sDAG with central edge $e$ introduced by the NNI.
However, this criterion does not yield an efficient algorithm, as we would need to enumerate all trees in $T_e$ and compute their likelihoods.

Our approach is to instead store local choices of subtrees, which may yield a sufficient approximation to the tree maximizing the likelihood above.
These local choices determine the topology and branch lengths, and so one can evaluate the likelihood of the tree in a classical way.
In this section we focus on giving an intuitive understanding of how these local choices work, and full details are found in the supplementary material.

\begin{figure}[h!]
\centering
\includegraphics[width=\textwidth]{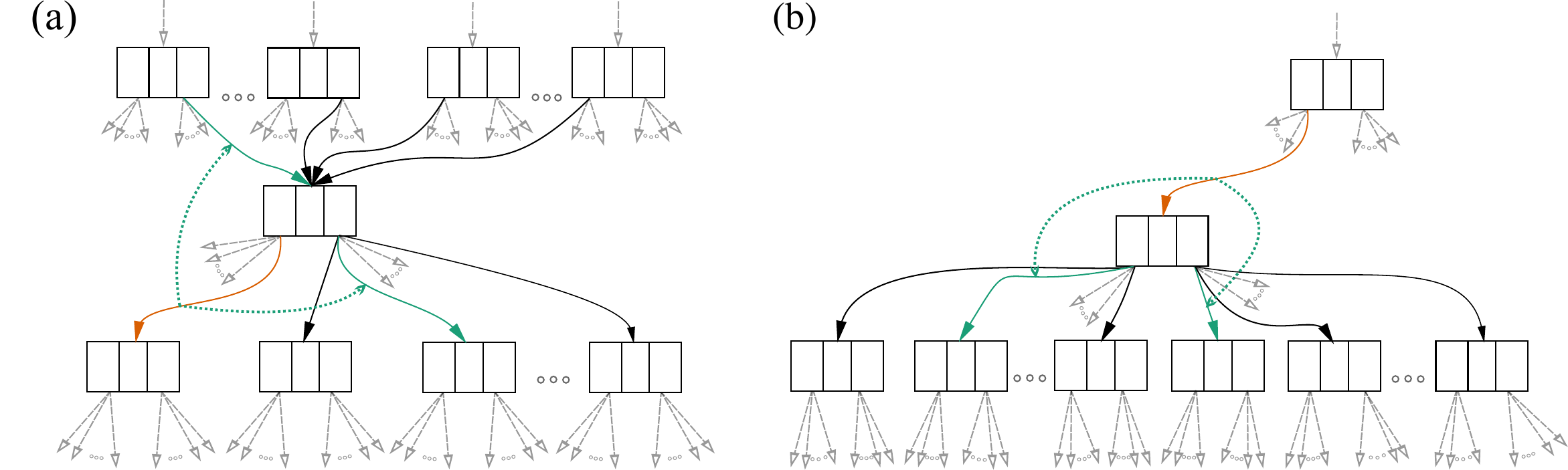}
\caption{An example of choice maps. 
In (a), the green solid edges are the parent and sibling edges returned by the rootward choice map (schematized by dashed line) for the given orange edge.
In (b), the green solid edges are the left and right child edges returned by leafward choice map for the given orange edge.
The ``best known tree'' for the given orange edge begins with the five highlighted edges: the orange edge common to both (a) and (b), the two solid green edges from (a), and the two solid green edges from (b).
}
\label{fig:inuitive_choices}
\end{figure}

The local choices of subtrees are implemented by a data structure we call a ``choice map'', which simply records a specific choice of neighboring edges for each edge of the sDAG (Figure \ref{fig:inuitive_choices}).
We have a ``rootward choice map'', which when given an edge returns a parent edge and sibling edge (assuming the given edge does not begin with the universal ancestor). 
We have a ``leafward choice map'', which when given an edge returns a left child edge and right child edge (assuming the given edge does not end with a leaf).

These choice maps can be recursively applied, defining tree structures as follows.
The leafward choice map defines a tree in the direction of the leaves from any given edge of the sDAG by recursively applying it until reaching the leaves.
To get a tree in the direction of the root from any given edge of the sDAG, one uses the rootward choice map to pick edges toward the root, and the leafward choice map to choose edges descending from these edges which haven't already been determined.
Combining these, we get a tree we call the ``best known tree'' for the edge.
Although it may not be the maximum likelihood tree in the sDAG containing the edge, by the way we construct the choice map (described below) we believe it should be a good approximation.

The following is a more precise description of the deterministic process to find the best known tree for a given edge.
\begin{enumerate}
\item Initialize a graph $G$ consisting of the single given edge. 
\item While $G$ is not a tree on the taxon set, we examine the edges of $G$. 
\begin{enumerate}
\item For each edge $e$ not ending in a leaf: if $G$ does not contain a left and right child edge of $e$, then we add the two children provided by the leafward choice map at $e$.
\item For each edge $e$ not originating at the universal ancestor: if $G$ does not contain a parent and sibling edge of $e$, then we add the parent and sibling edges provided by the rootward choice map at $e$.
\end{enumerate}
\end{enumerate}

Furthermore, we can perform this construction after performing an NNI on the sDAG\@.
If we cache partial likelihood vectors (PLVs) at the nodes of the sDAG, we can evaluate the likelihood of a best known tree in time that is constant in the number of leaves of the sDAG\@.

Next we explain which edges are selected for the choice maps.
Suppose we generate an sDAG from a list of trees ordered by likelihood (the highest likelihood tree is first).
To each edge of the sDAG, we assign the branch length from the first tree of the list containing the edge. 
Consider Figure \ref{fig:inuitive_choices}, where the given edge is highlighted in orange.
In panel (a), we select the parent and sister edges highlighted in green; the two child edges are selected and highlighted in panel (b).
The orange edge may appear in more than one tree of the original list, but since the trees are ordered by likelihood, we focus on the first (highest likelihood) tree with this edge.
The four neighboring edges are taken from this tree and the choices are recorded in the choice maps.
In particular, note that the choice maps depend on the likelihood-ordered list of trees used to construct the sDAG\@.

Next we will write out a complete example for an sDAG constructed from two input trees (Figure~\ref{fig:choiceMapInitialization}).
In this toy example we have trees $\tau_0$ and $\tau_1$ on seven taxa. 
We assume the classical setup for phylogenetic inference, with a sequence alignment and model, as well as branch lengths along the edges, so we can calculate the likelihood of a tree topology.
Assume that $\tau_0$ is of higher likelihood than $\tau_1$.
Figure \ref{fig:choiceMapInitialization}a depicts $\tau_0$, $\tau_1$, and the sDAG spanned by the two.
This sDAG contains two additional trees (Figure~\ref{fig:choiceMapInitialization}b).
The edges of the sDAG are labeled with the maximum likelihood input tree ($\tau_0$ or $\tau_1$) containing the edge.
To build the best known tree for an edge, which need not be one of the input trees, we begin with the given edge and attach the immediately neighboring edges from the $\tau_i$ of the label. 
For the attached neighboring edges that touch neither root nor leaf, we must choose two additional edges (either children or parent and sibling) to flesh out the best known tree. 
Such additional edges are taken from the input tree (again $\tau_0$ or $\tau_1$) given by the label of the attached edge. 
We continue this process until a tree is fully constructed.
The best known tree for the edges of the sDAG labeled $\tau_0$ is $\tau_0$.
The best known tree for edges labeled $\tau_1$ is $\tau_2$ for edges above the subsplit $\subsplit{1}{2,3,4,5}$ and $\tau_3$ for edges below.

\begin{figure}[th!]
\centering
\includegraphics[width=0.98\textwidth]{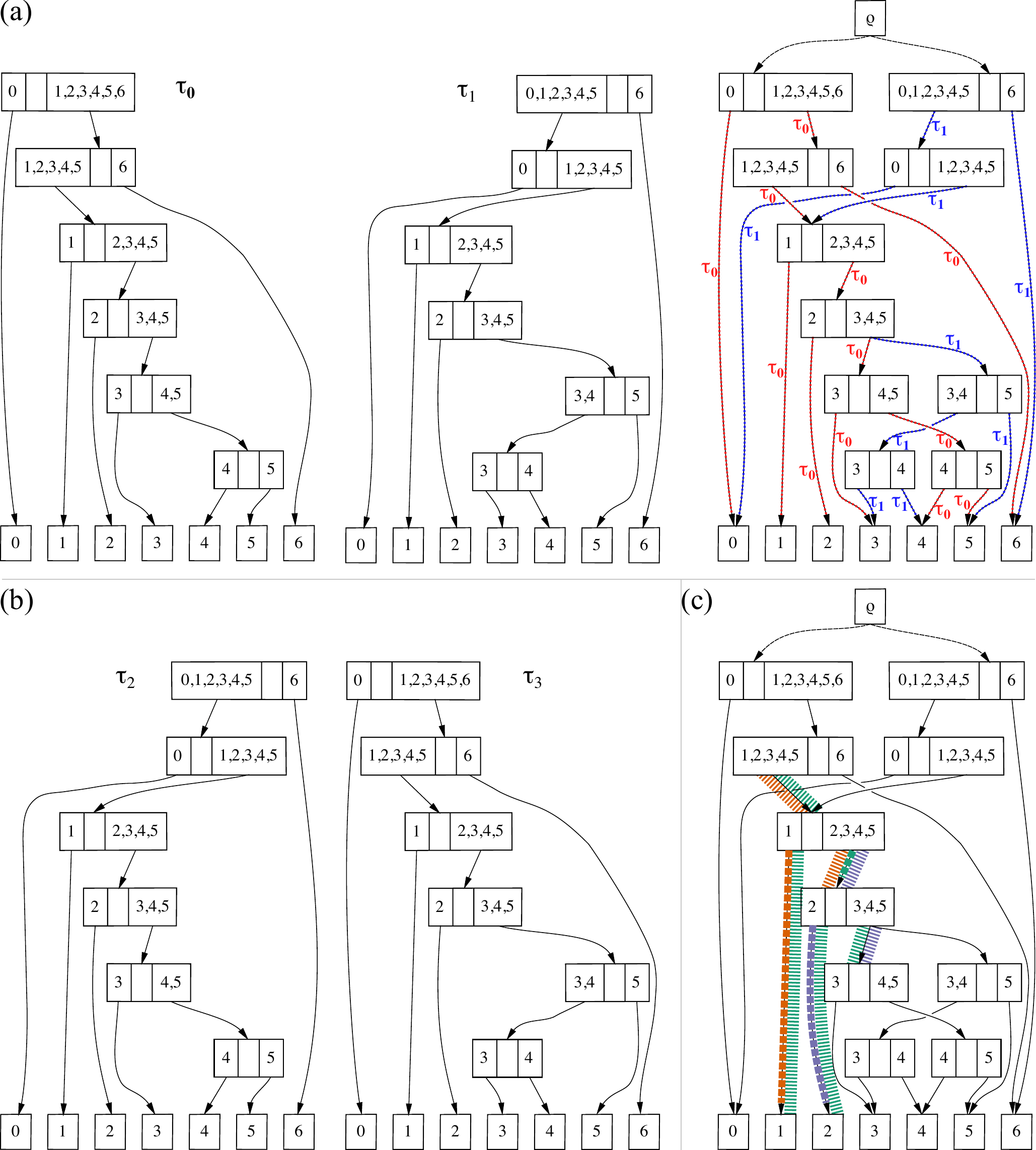}
\caption{Choice map initialization for an sDAG, assuming $\tau_0$ has higher likelihood than $\tau_1$.
In (a) are two trees on seven taxa and the sDAG built from the trees.
The edges of the sDAG are labeled with the topology, $\tau_0$ in red or $\tau_1$ in blue, where the edges first appeared.
In (b) are the two additional topologies in the sDAG.
In (c), we highlight the edges that have an option in the choice maps.
In orange is an edge (bold dashed line) and its chosen neighbors (lightweight dashed line), with a choice of parent edge;
in purple is an edge and its chosen neighbors, with a choice of right child edge;
and in green is an edge and its chosen neighbors, with a choice of parent edge and right child edge.}
\label{fig:choiceMapInitialization}
\end{figure}

For most edges in the sDAG of Figure \ref{fig:choiceMapInitialization}, there is only one choice for parent, sibling, and child edges.
The non-trivial choices can be phrased as: 
\begin{itemize}
\item Does $\pcsp{\subsplit{1}{2,3,4,5}}{\{1\}}$ use the parent edge from\\ $\subsplit{1,2,3,4,5}{6}$ or $\subsplit{0}{1,2,3,4,5}$?
\item Does $\pcsp{\subsplit{1}{2,3,4,5}}{\subsplit{2}{3,4,5}}$ use the parent edge from\\ $\subsplit{1,2,3,4,5}{6}$ or $\subsplit{0}{1,2,3,4,5}$?
\item Does $\pcsp{\subsplit{1}{2,3,4,5}}{\subsplit{2}{3,4,5}}$ use the right child edge to\\ $\subsplit{3}{4,5}$ or $\subsplit{3,4}{5}$?
\item Does $\pcsp{\subsplit{2}{3,4,5}}{\{2\}}$ use the sibling edge to\\ $\subsplit{3}{4,5}$ or $\subsplit{3,4}{5}$?
\end{itemize}
\sloppy
In Figure \ref{fig:choiceMapInitialization}c we highlight the edge $\pcsp{\subsplit{1}{2,3,4,5}}{\{1\}}$ in orange along with its chosen neighbor edges. Additionally, we highlight $\pcsp{\subsplit{1}{2,3,4,5}}{\subsplit{2}{3,4,5}}$ in green and $\pcsp{\subsplit{2}{3,4,5}}{\{2\}}$ in purple, as well as their chosen neighbor edges.

\fussy
The key point is each edge of the sDAG is assigned a tree systematically.  
Importantly, we can extend this to assign trees to the central edges of NNIs of the sDAG.
The details of extending the choice maps to such edges are given in the supplementary materials.
Furthermore, this assignment of trees allows for efficient phylogenetic likelihood calculations.
We define the ``top pruning likelihood'' of an NNI to be the likelihood of the best known tree for the central edge of the NNI.

The top pruning algorithm proceeds in the following manner. 
Suppose we are given a list of trees, ordered by likelihood.
\begin{enumerate}
\item Initialize $\sdag$ to the sDAG spanned by the topologies from the list of trees.
\item Assign branch lengths to the edges of $\sdag$ by taking the branch length from the first tree in the list containing a given edge.
\item Initialize choice maps for the edges of $\sdag$.
\item Add edges between all compatible parent and child subsplits. 
Assign choice maps with an $\argmax$ strategy (see Supplement, equation \eqref{eq:tpExtraEdgeChoices}). 
Assign branch lengths optimizing best known trees for these edges.
\item Create a list of NNIs, ordering by top pruning likelihood, for each edge in the sDAG (two NNIs per sDAG edge). 
We do not record NNIs in the list if they do not enlarge $\mathcal{D}$.
\item\label{item:tpLoopStart} Enlarge $\mathcal{D}$ to $\mathcal{D}^\prime$ with the highest top pruning likelihood NNI of the list.
	Assign choice maps and branch lengths as in \eqref{eq:topPruningChoices}.
\item Enlarge $\mathcal{D}^\prime$ to $\mathcal{D}^{\prime\prime}$ by adding all additional edges between compatible subsplits. 
	Assign choice maps and branch lengths as in \eqref{eq:tpExtraEdgeChoices}.
\item Remove the NNI of step \ref{item:tpLoopStart} from the list. 
\item\label{item:tpAddNNIs} Insert into the list, while maintaining the order by top pruning likelihood, the new NNIs that enlarge $\mathcal{D}^{\prime\prime}$.
The top pruning likelihood is not updated for NNIs already present in the list.
\item Return to step \ref{item:tpLoopStart} with $\mathcal{D}^{\prime\prime}$ in place of $\mathcal{D}$.
\end{enumerate}
The looping step of the algorithm repeats either for a fixed number of iterations or until the likelihood of all NNIs in the list are below a given threshold.

\subsubsection*{Generalized Pruning}

The generalized pruning (GP) objective applies the NNI to an sDAG that most increases a topology-marginal composite likelihood called the generalized pruning likelihood~\citep{GP}.
Specifically, let $\sdag_e$ be the sDAG from applying an NNI with central edge $e$.
Let $T_e$ be the set of trees of $\sdag_e$ that contain $e$. 
The GP algorithm applies the NNI that maximizes
\begin{gather}\label{eq:gpLikelihood}
\prod_{j=1}^K\sum_{\tau\in T_e} p_{\psi}(\bY^j\mid \tau) \, p(\tau\mid e),
\end{gather}
across edges $e$. 
The term $p(\tau\mid e)$ is the relative prior probability of $\tau$ among the topologies of $T_e$. 
Explicitly, 
\begin{gather*}
p(\tau\mid e) = \frac{p(\tau)}{\sum_{\tau^\prime\in T_e} p(\tau^\prime)},
\end{gather*}
where $p(\tau)$ is a prior distribution on topologies. 
Our implementation uses the uniform distribution, so that $p(\tau\mid e)=\frac{1}{|T_e|}$.
The likelihood in \eqref{eq:gpLikelihood} is the \emph{per‑edge marginal likelihood} introduced in \cite{GP}, the details of which are given in the section of the same name. Specifically, this is the generalized pruning per-edge composite marginal likelihood of the edge $e$ in the sDAG after applying the NNI. We call this likelihood the generalized pruning likelihood of the NNI.

Note that we call \eqref{eq:gpLikelihood} a ``composite likelihood'' because it is taking a product of per-site marginal likelihoods rather than marginalizing over the complete likelihood.
This correct marginal likelihood of the edge would be
\begin{gather*}
\sum_{\tau\in T_e} p(\tau\mid e)  \prod_{j=1}^K p_{\psi}(\bY^j\mid \tau).
\end{gather*}
However, there are no efficient means known of directly calculating this correct marginal likelihood, whereas the generalized pruning version scales linearly in the size of the sDAG\@.
The generalized pruning marginal likelihood is computed with a traversal of the sDAG that populates partial likelihood vectors at each node, where each node has a PLV for each site of the sequence alignment.
In contrast, to compute the true marginal likelihood, each node would require a PLV not just for each site, but each site and distinct topology with that node. 
That is to say, the GP likelihood is computed using PLVs that are shared by topologies, while the true marginal likelihood would require separate PLVs for topologies.
Although it is not the same as the true marginal likelihood, we have found that the GP likelihood is a sufficient approximation of the true likelihood for the purpose of optimizing branch lengths~\cite{GP}.

Suppose we are given a list of starting trees ordered by likelihood, the generalized pruning systematic inference algorithm is as follows.
\begin{enumerate}
\item Initialize $\sdag$ to the sDAG spanned by the topologies from the list of trees.
\item Assign branch lengths to the edges of $\sdag$ by taking the branch length from the first tree in the list containing a given edge.
 Optionally, we further optimize the branch lengths to maximize the overall generalized pruning likelihood of the sDAG.
\item Add edges between all compatible parent and child subsplits; assign GP per-edge composite marginal likelihood optimized branch lengths to these edges.
\item\label{item:gpLikelihood} Create a list of NNIs, ordering by generalized pruning likelihood, for each edge in the sDAG (two NNIs per sDAG edge). 
We do not record NNIs in the list if they do not enlarge $\mathcal{D}$.
\item\label{item:gpLoopStart} Enlarge $\mathcal{D}$ to $\mathcal{D}^\prime$ with the highest GP likelihood NNI of the list.
Add any additional edges between either of the new subsplits and compatible existing subsplits. 
Assign GP per-edge composite marginal likelihood optimized branch lengths to the new edges. 
\item Remove the NNI of step \ref{item:gpLoopStart} from the list. 
\item\label{item:gpAddNNIs} Insert into the list, while maintaining the order by GP likelihood, the new NNIs that enlarge $\mathcal{D}^{\prime\prime}$ .
The GP likelihood is not updated for NNIs already present in the list.
\item Return to step \ref{item:gpLoopStart} with $\mathcal{D}^\prime$ in place of $\mathcal{D}$.
\end{enumerate}
The looping step of the algorithm repeats either for a fixed number of iterations or until the GP likelihood of all NNIs in the list are below a given threshold.
The GP likelihoods in items \ref{item:gpLikelihood} and \ref{item:gpAddNNIs} are calculated after optimizing branch lengths.

\subsection*{Implementation of systematic inference}

The necessary functionality for both NNI-searches are implemented in the Python-interface C++ library \bito (\url{https://github.com/phylovi/bito}) and an interface to perform a search is further implemented in Python (\url{https://github.com/matsengrp/sdag-nni-experiments}).
Both top pruning and generalized pruning use PLVs for fast and efficient likelihood calculations.
The PLVs for top pruning are defined as usual for a two-pass version of Felsenstein's pruning algorithm and are propagated along the choice maps. 
The PLVs for generalized pruning follow a different pattern and are discussed in detail in \cite{GP}.
It is these likelihoods and associated PLVs that dominate the computational expense of our algorithms, while maintaining the the graph structure of the current sDAG and potential additions is secondary.

When new edges are introduced for top or generalized pruning, the optimization of associated branch lengths is done as follows.
We take one branch length, hold the others fixed, apply Brent optimization (maximizing the likelihood of the edge in terms of best known tree likelihood or generalized pruning), repeat with another branch length, and continue until values have approximately converged.
We leave the branch lengths of old edges unaltered.
This approach showed good performance for branch length optimization in previous work~\cite{GP}.

\subsection*{Benchmarking setup}
As in previous work (e.g., \cite{Lakner2008, Hhna2012-pm, Larget2013-uo, whidden2015, Fourment2020-te}), we compare our inferences to very long \mrbayes\ ``golden runs''. 
We use common benchmark data sets, which we call the DS-datasets, as well as a 100 influenza A (flu) sequence data set from \cite{vbpi2}.
Each very long run of \mrbayes\ yields numerous topologies and posterior density estimates of these topologies, which form the empirical posterior.
We take the 95\% credible set to be the minimal size set of topologies whose cumulative density is at least 95\%.
Basic information for the data sets and their empirical posteriors is in Table \ref{table:dataSetStats}.
We say a subsplit is credible if it appears in at least one topology of the 95\% credible set. 
This is different from taking a credible set of subsplits from a topology marginalized probability on subsplits.

\begin{table}\centering
\begin{tabular}{ccccc}
&&& 95\% credible set & empirical posterior
\\
data set	& taxa & sites & topology count & topology count
\\\hline
1 & 27 & 1949 & 42 & 1245
\\\hline
3 & 36 & 1812 & 16 & 240
\\\hline
4 & 41 & 1137 & 219 & 4539
\\\hline
5 & 50 & 378 & 260894 & 298768
\\\hline
6 & 50 & 1133 & 157942 & 195816
\\\hline
7 & 59 & 1824 & 756 & 6000 
\\\hline 
8 & 64 & 1008 & 4329 & 26442
\\\hline
flu & 100 & 1681 & 16475 & 20262
\end{tabular}
\caption{Description of the DS-datasets.}
\label{table:dataSetStats}
\end{table}

We will use the term ``diffuse'' to qualitatively describe the size of the posterior credible set.
Taking the numbers in Table 1 (either for credible topologies or posterior topologies) the order of data sets from least to most diffuse is DS3, DS1, DS4, DS7, DS8, DS6, and DS5.
The data sets DS1, DS3, DS4, and DS7 are not diffuse, while DS5 and DS6 are very diffuse, and DS8 is somewhere between.
While there is another data set, commonly called DS2, the posterior is only a few topologies and so we ignore it here.

For each data set, we calculate how much of the posterior density is captured per NNI-search iteration and by run time.
We use short runs of \mrbayes\ for a direct comparison of run time.
These short runs of \mrbayes\ use similar specifications to the golden runs, except we record all topologies (i.e., sample frequency 1 and no burn-in).
With these short runs, we have two comparisons with an NNI-search.
First, we compare the cumulative posterior density of the topologies from an NNI-search to that of the topologies from the short run.
While this is a fair comparison in terms of run time, comparing the density from topologies of an sDAG to topologies from a simple list puts the short runs at a disadvantage.
Our second comparison addresses this issue by comparing the density of topologies from an NNI-search to the density of the topologies in the sDAG spanned by the topologies from the short run of \mrbayes.
This gives a more accurate comparison in terms of cumulative density, as we compare an sDAG to an sDAG.
We emphasize that while our systematic inference algorithms produce sDAGs that represent trees, we ignore branch lengths and work with topologies for these comparisons.
While we can efficiently build an sDAG from a large list of topologies, there is required processing of topologies from \mrbayes\ before building the sDAG (rerooting, ordering of splits of subtopologies in the newick string, etc.).
When memory is not a constraint, this processing of topologies is also quick and efficient. 
However, for a very large number of topologies (as is the case with DS5), we use a less efficient method.
The run-time we report for the sDAG spanned by topologies of \mrbayes\ includes the run-time for \mrbayes, the estimated run-time of processing topologies with the efficient implementation, and the run-time for constructing the sDAG from the processed topologies.
As discussed in the following section, generalized pruning performs poorly and so we do not show it in every benchmark.

We start all of the searches (top pruning, generalized pruning, and short runs of \mrbayes) with the maximum posterior density topology of the golden run with branch lengths optimized by \iqtree\ for likelihood. 
Starting the searches at the maximum posterior density topology is a best-case scenario for performance, as the search starts at the highest peak of the distribution.

We also compare top pruning and the short runs of \mrbayes\ by the quality of subsplits. 
We calculate both how many of the credible subsplits are found by a search and how many of the subsplits found by a search are credible. 
The comparison with top pruning and short runs of \mrbayes\ matches the two on the number of subsplits encountered.
In the ideal setting, the search methods would only add subsplits present in the posterior.

Our search methods can take multiple trees as input, so we investigate two additional starting points.
The first is a ``best case'' minor variation: we begin the search with the top 10 posterior density topologies with likelihood optimized branch lengths. 
The second is a more realistic setting and randomized: we take the unique topologies produced by 200 runs of \raxml~\cite{raxml} with likelihood optimized branch lengths.
We find that \raxml\ produces more distinct topologies than \iqtree\ when given random starting trees (data not shown).
We compare the performance with different starting trees by cumulative density per iteration.
There is not a comparison with short runs of \mrbayes.
It is worth noting that we may use the trees from \raxml\ without any knowledge of the credible set or an empirical posterior distribution.

For consistency, all timing scripts were run with the \slurm\ job scheduler with exclusive access to a single node.
The experiments of this section can be recreated by following the instructions in the experiments GitHub repository (\url{https://github.com/matsengrp/sdag-nni-experiments}).

\section*{Results}

We first describe the disappointing performance of generalized pruning, and then the somewhat better performance of top-pruning.
For top-pruning we also provide a more detailed analysis, discussing run time, performance with multiple starting points, and performance on the larger flu data set.
For this first set of analyses, we will be interested in the total posterior density of the topologies present in the sDAG generated by the various methods through time.

The performance of generalized pruning is lackluster. 
Regardless of data set, generalized pruning fails to capture a reasonable amount of posterior density in a reasonable amount of time (Figure \ref{fig:generalizedPruningFoundPosterior1} for DS1 and DS3-6, Figure \ref{fig:generalizedPruningFoundPosterior2} for DS7-8).
It is not competitive with the list of topologies from the short runs of \mrbayes.
Even on DS3, the least diffusive data set, generalized pruning spends many iterations adding subsplits and edges not contributing to the density.

\begin{figure}[!t]
\includegraphics[width=0.90\textwidth]{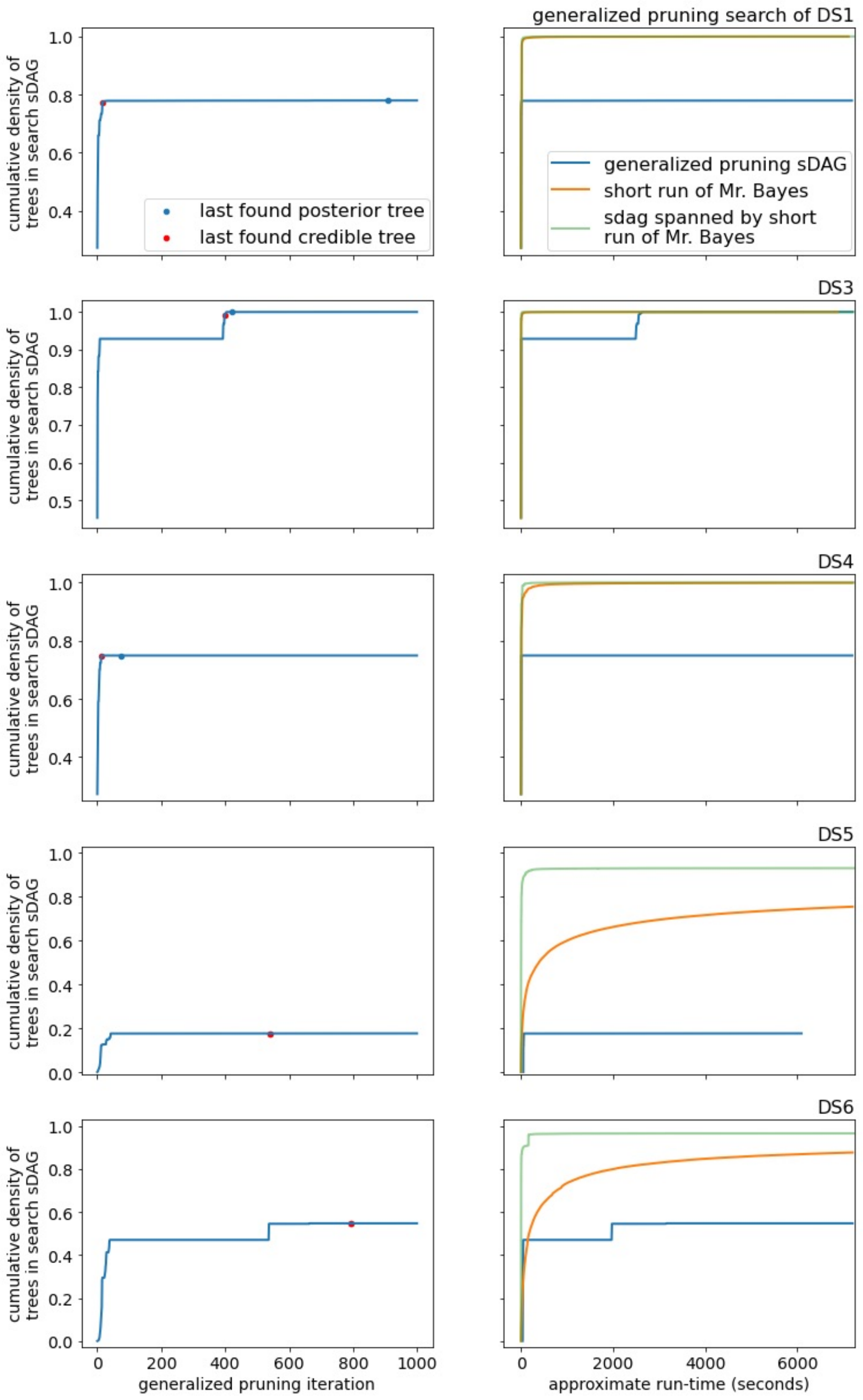}
\caption{The empirical posterior density found by generalized pruning on a subset of the DS-datasets and a comparison with MCMC.
The red dots in the plots of the left indicate the last iteration finding a topology of the credible set and the blue dots indicate the last iteration finding a topology of the empirical posterior.  
}
\label{fig:generalizedPruningFoundPosterior1}
\end{figure}

Top pruning fares better than generalized pruning on the DS-datasets.
It consistently outperforms the topologies of a short MCMC run and is competitive with the sDAG spanned by these topologies.
While top pruning experiences long lengths of time with little to no increase in coverage of the posterior density, it usually finds its way back and captures much of the posterior density (Figure \ref{fig:topPruningFoundPosterior1} for DS1 and DS3-6, Figure \ref{fig:topPruningFoundPosterior2} for DS7-8).

The credibility of subsplits found by top-pruning tends to be worse than the short runs of \mrbayes\ (recall that we say that a subsplit is credible if it appears in at least one topology of the 95\% credible set).
Top pruning and \mrbayes\ are very similar on DS3 and DS7, but on the other datasets top pruning adds many non-credible subsplits along with credible subsplits (Figure \ref{fig:tp_uniform_prior_sdag_quality}).
The orange (magenta) lines are essentially the true-positive rates of top pruning (short MCMC, respectively) in terms of subsplits.
While not included here, we performed the same analysis with sDAG edges rather than subsplits and the results were similar.

Multiple starting trees are beneficial, but the benefit depends on the trees. 
Using the 10 highest probability trees from the credible set, or high likelihood trees generated by \raxml, of course, provides the gain in cumulative density of the additional topologies (Figure \ref{fig:tp_multiple_start}).
Past that, performance improvements are more subtle.
Suppose we take the sDAG formed by the multiple initial trees and compare with the sDAG from running top pruning, without the additional trees, until we have approximately the same cumulative density of topologies in the sDAGs.
After applying a few iterations of top pruning to both of these sDAGs, one may have a higher cumulative density of topologies than the other.
Using the 10 highest probability trees usually does not give this kind of improvement, while using the high likelihood \raxml\ trees often does.
That is, using the trees from \raxml\ improves the selection of NNIs toward the beginning of the search, while using the 10 highest probability trees does not.
One possible explanation is the top 10 trees may be too close to each other in the posterior distribution and top pruning would visit these topologies fairly quickly and approximately in order, while the trees from \raxml\ may be near different peaks of the distribution.

Our implementation of top pruning does not exhibit linear run time with the number of iterations (Figure \ref{fig:topPruningRunTimeDS4}). 
This is important to consider, since we compare top pruning against MCMC by run-time. 
Later iterations of top pruning may add more edges than in previous iterations. 
This means we expect more likelihood calculations on average in later iterations than in early iterations. 
With this, we do not expect linear run time is possible, but it is at worst quadratic.
There is also the issue of scaling in terms of the number of taxa (the number of leaves in the sDAG, the number of sequences in the alignment, etc.).
The performance of top pruning on the flu data set is consistent with that of the DS-datasets (Figure \ref{fig:tp_flu100}).
This suggests top pruning may scale well with the number of taxa.

Lastly, we compare the empirical posterior probability of an edge to the likelihood of the best known tree for the edge (the top pruning likelihood). 
By posterior probability of an edge, we mean the topology marginalized probability, which is the sum of posterior probabilities of topologies containing the edge. 
Across data sets, high posterior probability corresponds well to high likelihood of the best known tree of an edge
(Figure \ref{fig:posterior_vs_tp}). 
This suggests that the challenge with the top pruning algorithm is its tendency to find and include edges outside of the credible set of topologies, but top pruning does a fair job of ranking edges within the credible set.

\vspace{-10ex}
\begin{figure}[!b]\centering
\includegraphics[width=0.90\textwidth]{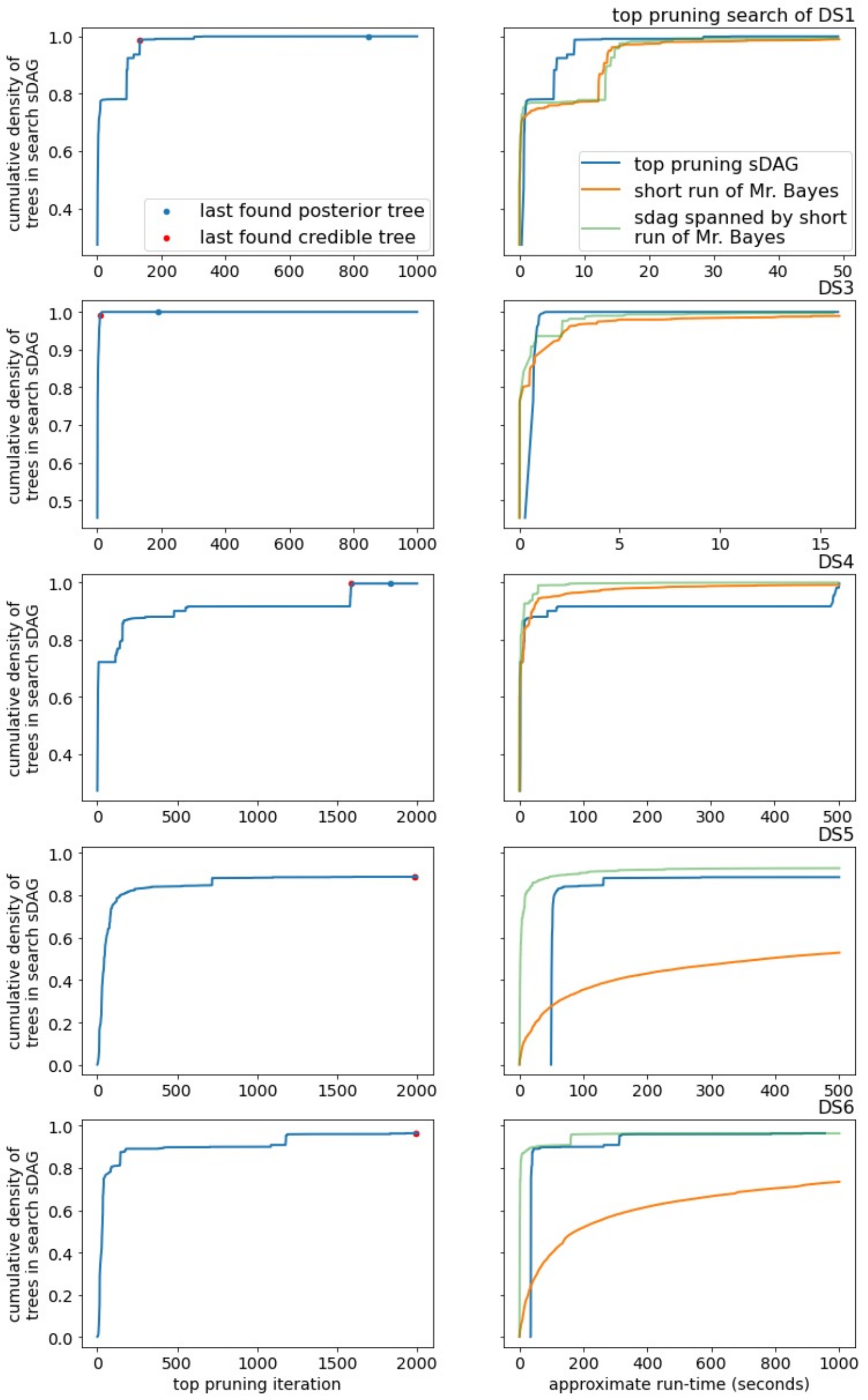}
\caption{The empirical posterior density found by top pruning on a subset of the DS-datasets and a comparison with MCMC.
Note the $x$-axis scale varies between data sets.
}
\label{fig:topPruningFoundPosterior1}
\end{figure}

\begin{figure}[!h]\centering
	\includegraphics[width=0.98\textwidth]{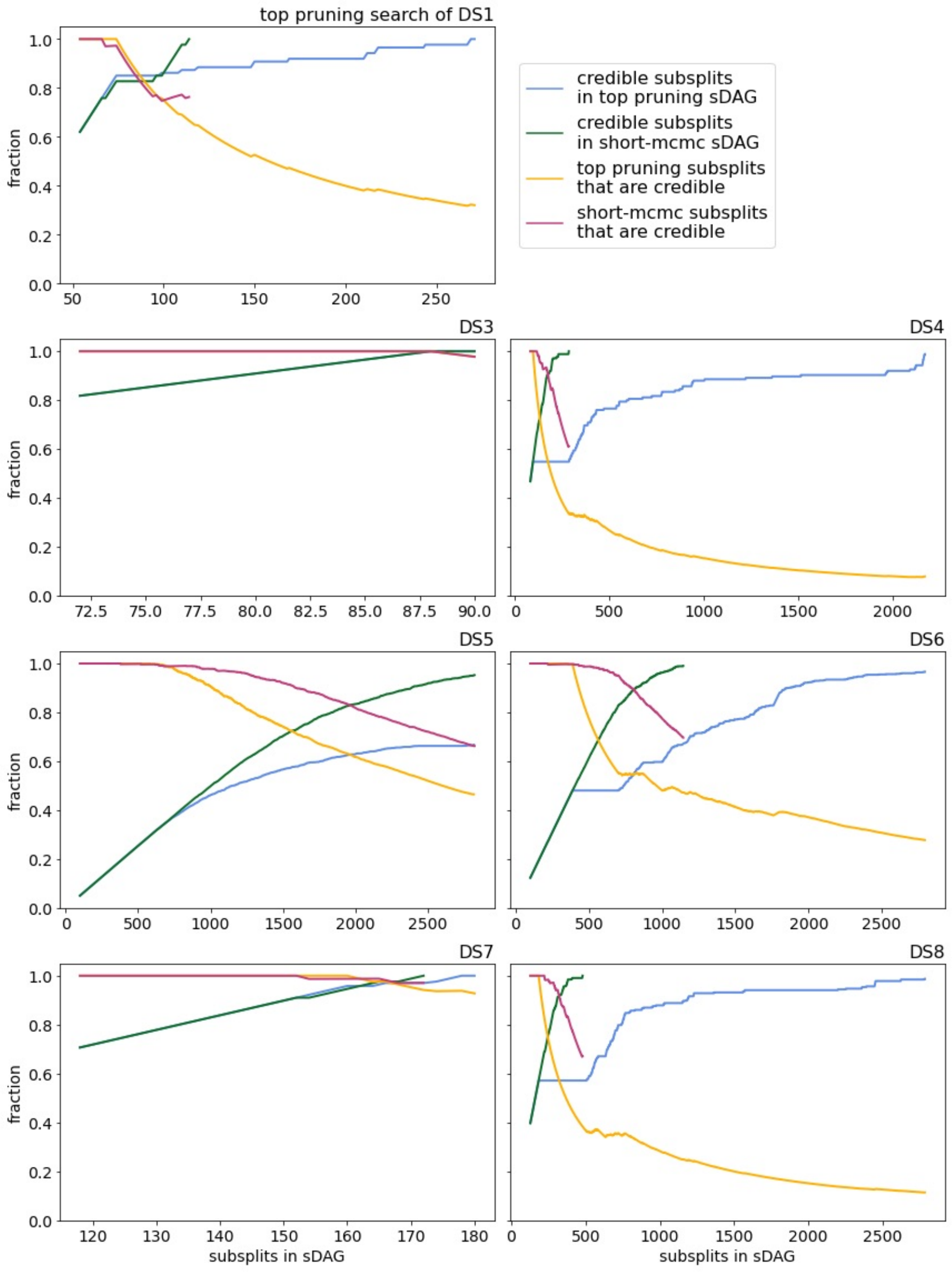}
\caption{Quality of sDAGs generated by top-pruning and short-MCMC in terms of subsplits.
For each data set, two sDAGs were iteratively enlarged (by top pruning or aggregating trees from \mrbayes).
On the $x$-axis we have the number of subsplits in the sDAGs and on the $y$-axis we have
the fraction of credible subsplits in the sDAGs and the fraction of subsplits in the sDAGs that are credible.
In all cases, the initial sDAGs are given by the maximum posterior density topology, and so all subsplits of the initial sDAGs are credible.
}
\label{fig:tp_uniform_prior_sdag_quality}
\end{figure}

\begin{figure}[!h]\centering
	\includegraphics[width=0.98\textwidth]{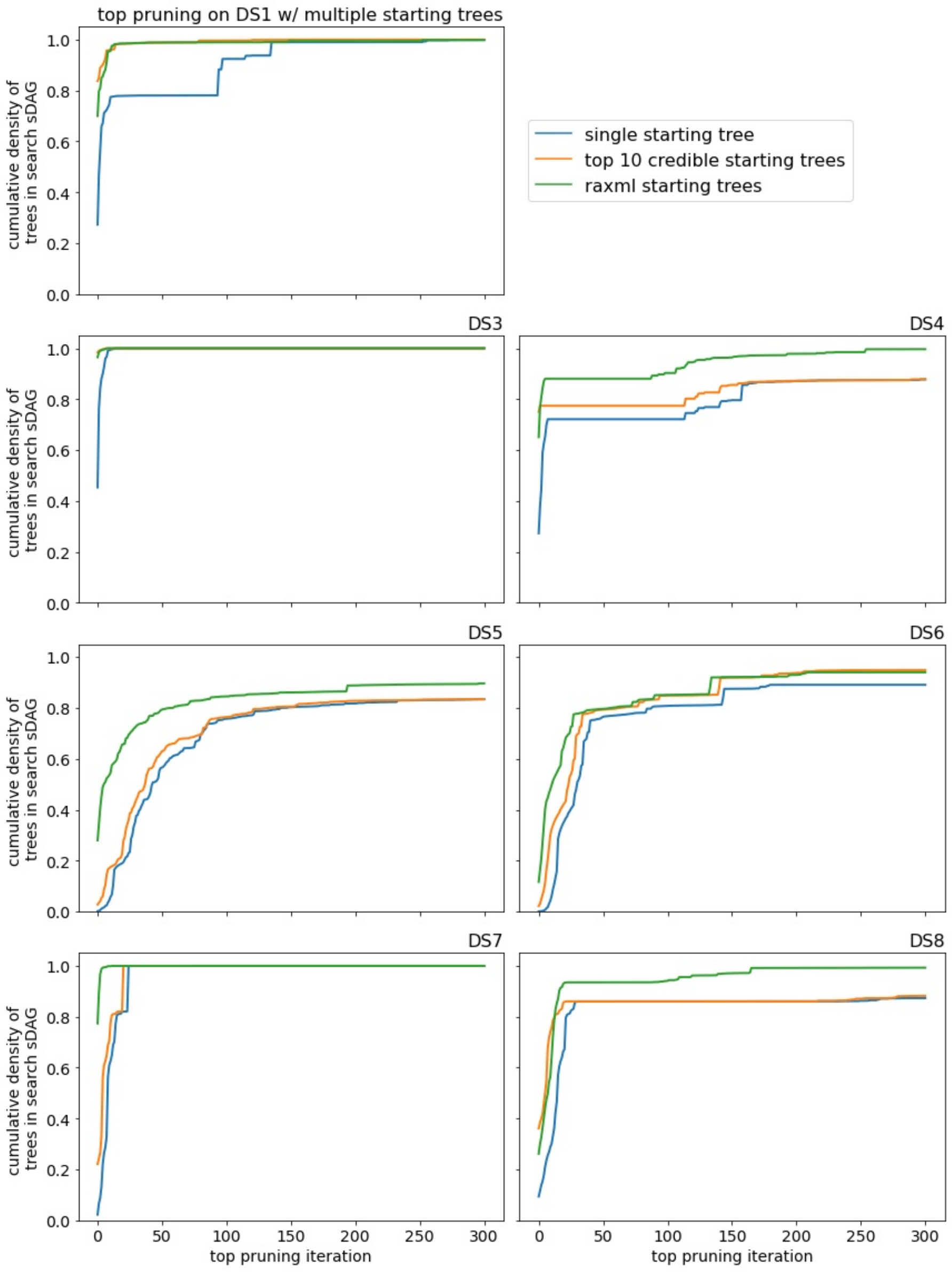}
	\caption{Performance of top pruning with different starting points.}
	\label{fig:tp_multiple_start}
\end{figure}
\clearpage

\begin{figure}[!t]\centering
	\includegraphics[width=0.98\textwidth]{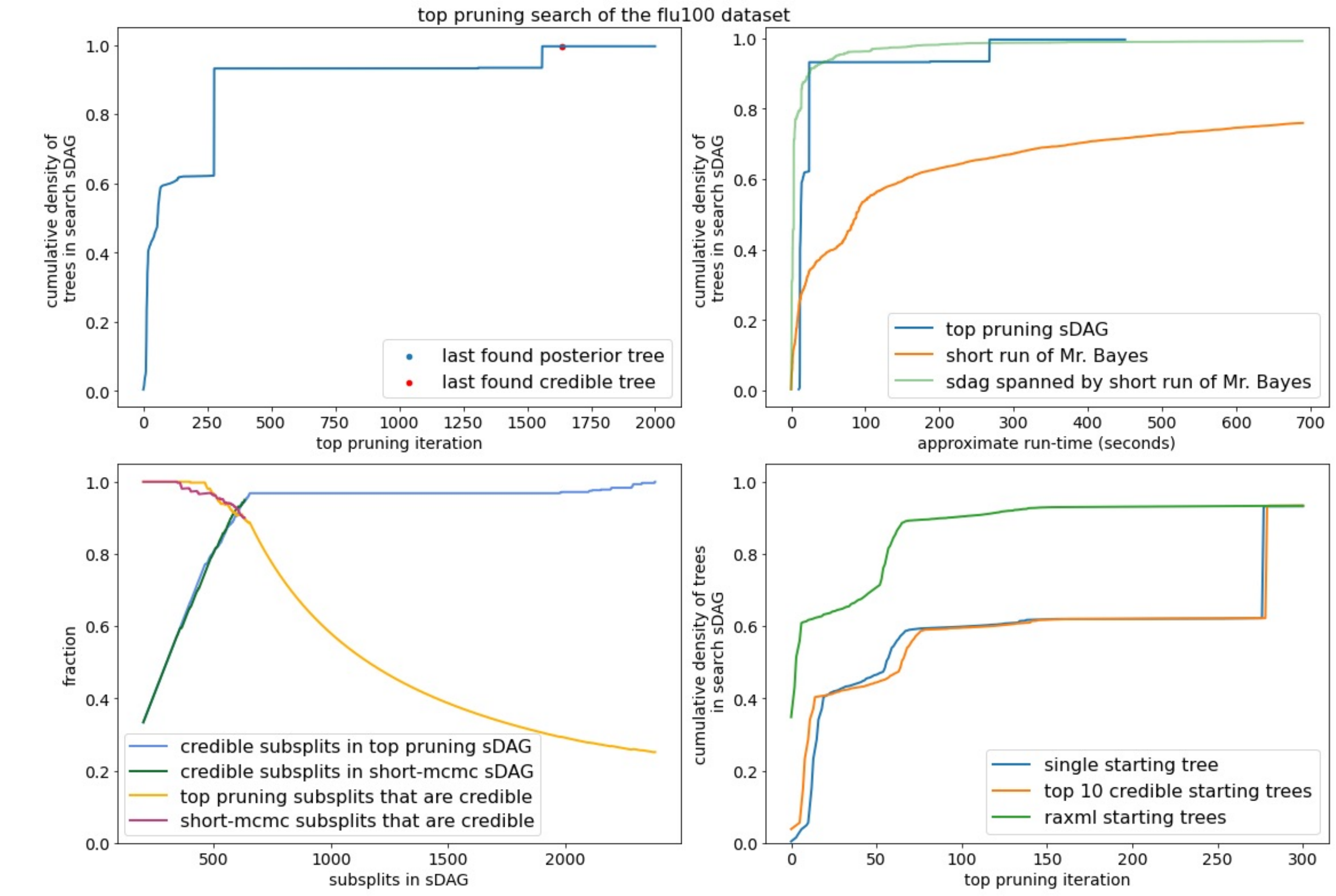}
	\caption{Performance of top pruning on the flu100 data set.}
	\label{fig:tp_flu100}
\end{figure}

\begin{figure}[!h]\centering
\includegraphics[width=0.9\textwidth]{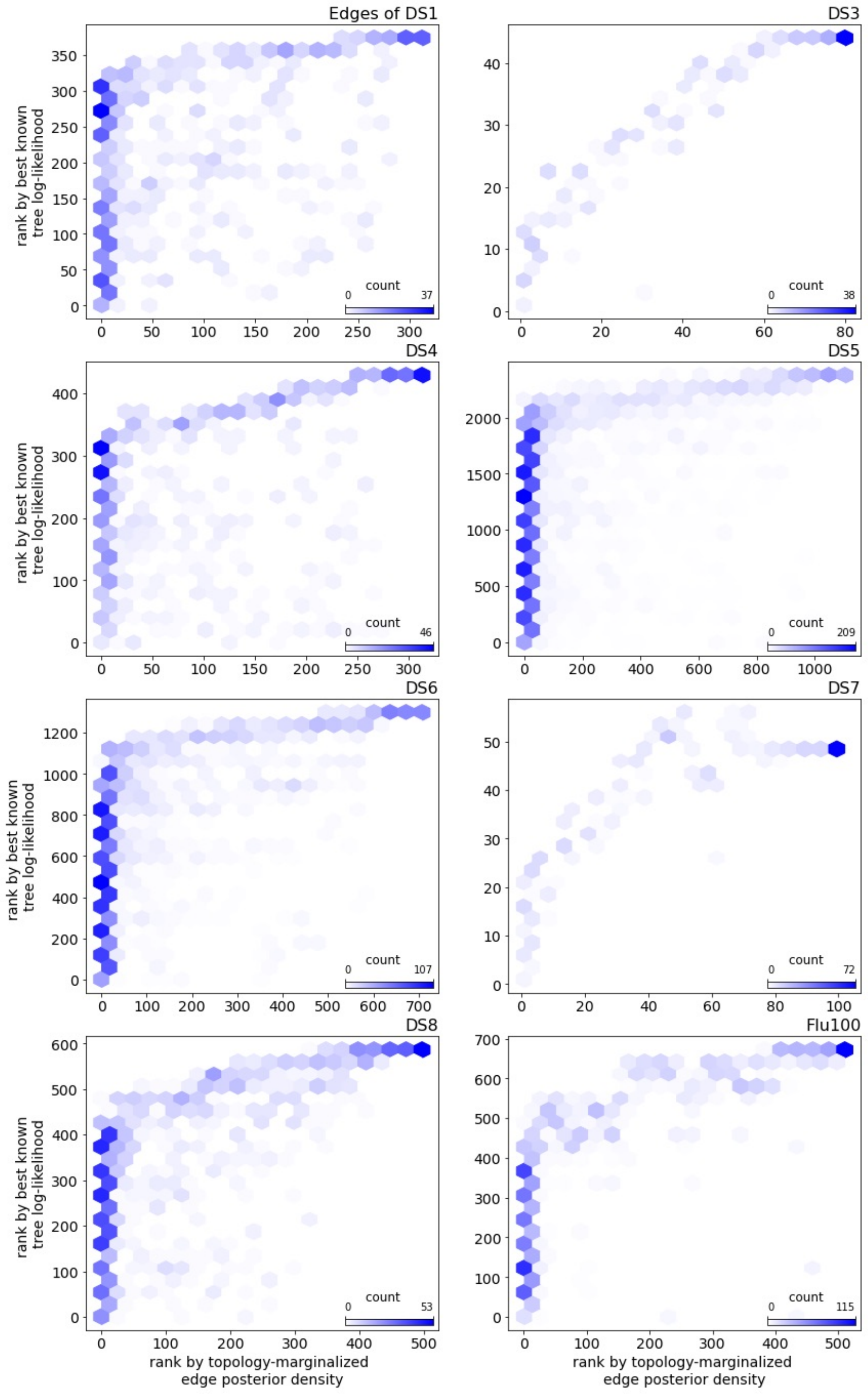}
\vspace{-1ex}
\caption{
A comparison of edges ranked by posterior probability and top pruning likelihood.
The edges are ordered by their posterior probability along the $x$-axis (with the lowest probability edge at $0$) and by their top pruning likelihood along the $y$-axis (with the lowest likelihood edge at $0$). 
Edges are binned for counts, with a darker color meaning a higher count.
}
\label{fig:posterior_vs_tp}
\end{figure}

\clearpage

\section*{Discussion}

In this paper we develop strategies to identify potentially high posterior density regions of tree space.
This fits as part of a larger program to infer Bayesian phylogenetic posterior distributions via optimization~\cite{sbn, vbpi,vbpi2, GP}.
We emphasize that our goal for this paper is to find high-posterior regions of topologies, rather than trying to determine posterior densities of these topologies.
Once a region is located by a search method, be it a systematic inference or aggregation of trees, another method such as variational inference may be used to infer posterior densities in this region.
Such methods may additionally infer distributions for branch lengths and evolutionary process parameters, allowing for joint posterior density estimates.
One could also imagine using sDAG inference to build a means of estimating split frequency as a recent development of adaptive MCMC~\cite{Meyer2021-xl}.

While top pruning performs well on our test data sets, the number of sequences is relatively small, with 100 being the largest. 
We do not yet know how well top pruning scales to larger data sets, where we have thousands or tens of thousands of sequences.
Our current method of benchmarking is not possible for such data sets, as we require an empirical estimate of the true posterior.
However, since top pruning lags behind the performance of aggregating MCMC-sampled trees into an sDAG on these small datasets, that seems like a more promising approach for future work and benchmarking.

Generalized pruning demonstrates poor performance with an NNI search, yet it has been shown useful when determining branch lengths \cite{GP}. 
We have not found any way to improve the performance here. 
Perhaps the generalized pruning likelihood of an edge does not sufficiently capture the quality of the newly introduced topologies.

One could imagine extensions allowing more general additions to the sDAG in our search algorithms.
When proposing and ranking NNIs for top pruning, we could try edges other than those present in the choice maps.
For example, rather than using the parent of $t$ specified in the choice map, we could try the other edges to $t$ as candidate parents for $t^{\prime}$ (such as the dotted blue edge in Figure \ref{fig:topPruningChoices}).
This requires only a little extra evaluation and branch length optimization.

Furthermore, when we apply the best NNI to the sDAG and remove it from the ranked list, the best known trees for some of the NNIs in the list may now have an edge present in the current sDAG that was not present before.
Such an edge has a branch length and choice maps that may yield a different likelihood from before.
In such cases, we could trigger re-optimization of branch lengths for that NNI and replace the corresponding likelihoods, PLVs, and choice maps.

An alternative version of top pruning adds exactly five new edges of the best known tree associated to an NNI.
That is, in each iteration, there are at most five edges of the best known tree for an NNI that are new to the sDAG.
We can add only these edges to the sDAG, rather than edges between all compatible subsplits or the five types of edges detailed in the subsection on Performing NNIs to the subsplit DAG.
This algorithm would enjoy linear run time, which the current implementation does not. 
However, this approach has a fundamental flaw: there are unobtainable subsplits and edges.
One can begin with a single topology on a set of taxa, iteratively build an sDAG by adding the five edges from the best known tree for an NNI, continue until all NNIs are exhausted, and obtain an sDAG missing valid subsplits and edges.
This is in strict contrast to NNIs on topologies, where every topology can be reached from any other topology by a sequence of NNIs.

Overall, our goal in this work is to do systematic inference to find high posterior density regions of tree space.
We introduce the first structures and operations on those structures to make this possible without considering each tree individually.
In order to do systematic inference, one needs a means of evaluating these structures, and here we introduce top pruning and generalized pruning.
Although top pruning shows reasonable performance, further work will be needed to improve performance over aggregating trees from an MCMC sampler into an sDAG\@.

\section*{Acknowledgements}

This work was supported through US National Institutes of Health grant AI162611.
Scientific Computing Infrastructure at Fred Hutch was funded by ORIP grant S10OD028685.
Dr.\ Matsen is an Investigator of the Howard Hughes Medical Institute.

\bibliographystyle{abbrv}
\bibliography{main}
\clearpage

\section*{Supplementary Materials}
\beginsupplement%

\subsection*{Additional Benchmarking Figures}

\vspace{-3ex}
\begin{figure}[!h]
\includegraphics[width=0.90\textwidth]{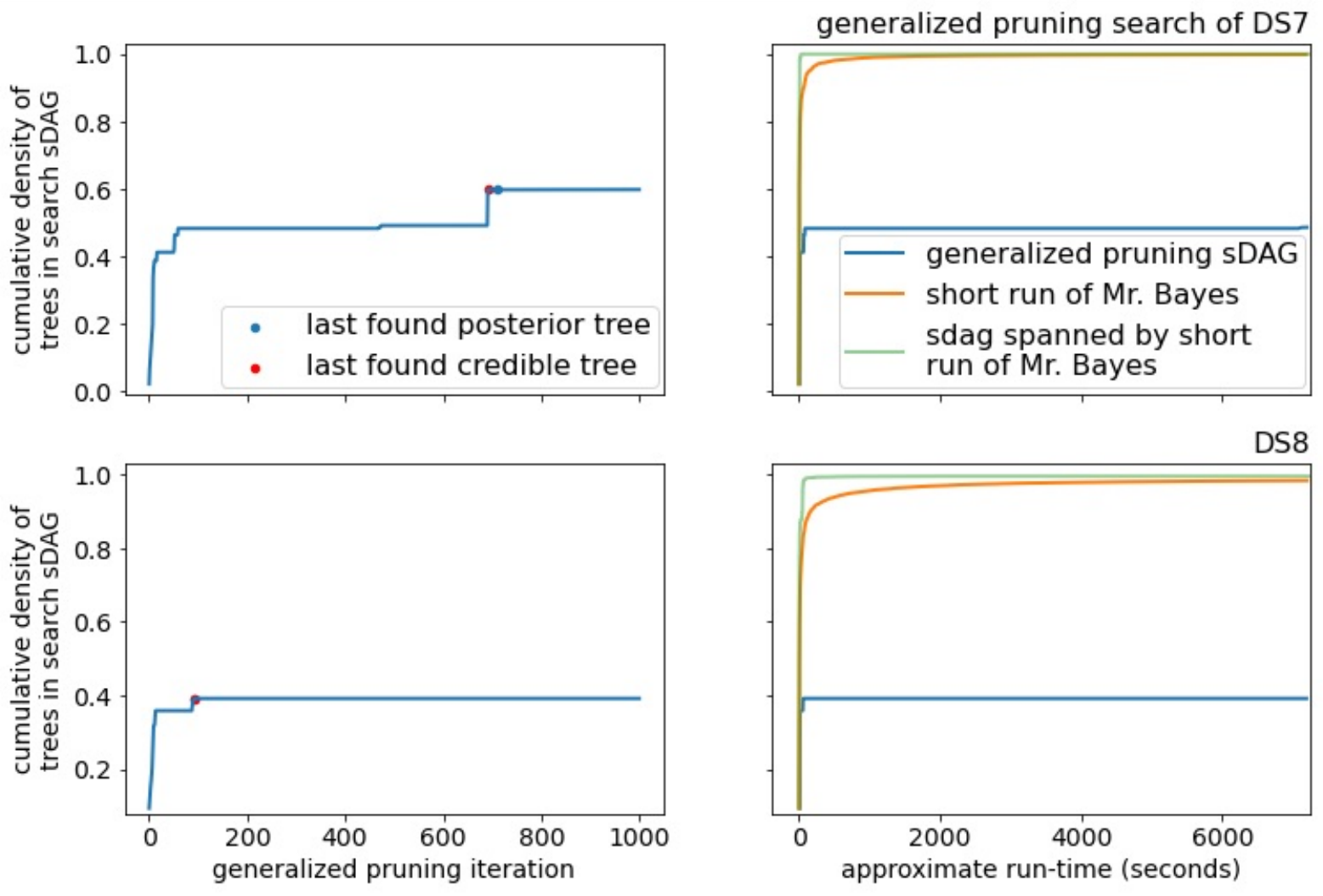}
\caption{The empirical posterior density found by generalized pruning on the remaining DS-datasets and a comparison with MCMC.
}
\label{fig:generalizedPruningFoundPosterior2}
\end{figure}

\vspace{-4ex}
\begin{figure}[!h]
\includegraphics[width=0.90\textwidth]{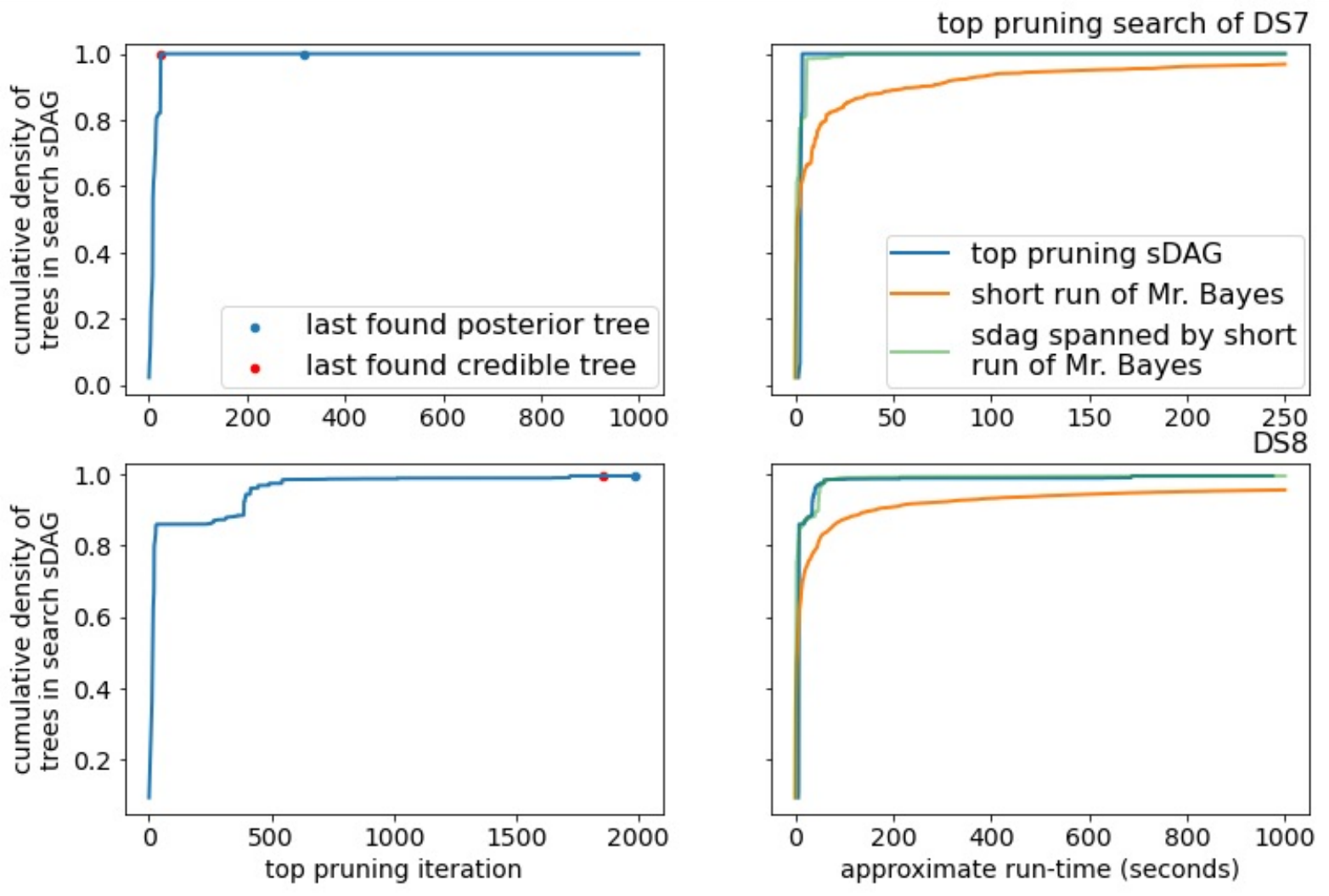}
\caption{The empirical posterior density found by top pruning on the remaining DS-datasets and a comparison with MCMC.
Note the $x$-axis scale varies between data sets.
}
\label{fig:topPruningFoundPosterior2}
\end{figure}

\begin{figure}[!t]\centering
\includegraphics[scale=0.4]{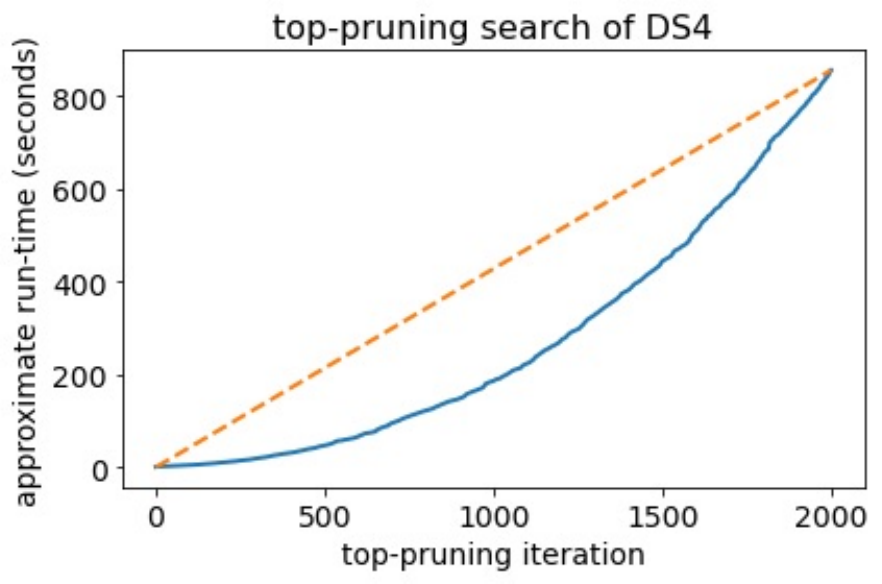}
\caption{The run-time (blue) of top pruning on DS4. 
The straight dotted orange line visually confirms convexity. 
All other data sets exhibit similar run-time behavior.}
\label{fig:topPruningRunTimeDS4}
\end{figure}

\subsection*{Compatible Subsplits and Edges}
At the end of the subsection on Performing NNIs to the subsplit DAG, we stated our preference for sDAGs with edges between all compatible subsplits and some properties of such sDAGs. 
Having described the top pruning algorithm, we now explain these preferences.
Consider the likelihood in \eqref{eq:fullTopLikelihood}, for which the top pruning likelihood serves as a proxy. 
For an sDAG without edges between all compatible subsplits, we would need to consider topologies past those that contain the central edge. 
It is possible for a new topology to exist in the post-NNI sDAG, not contain the central edge, yet have the maximum likelihood.
Even among topologies containing the central edge, the one of maximum likelihood may not be an NNI of a topology in the pre-NNI sDAG, making it impossible for a best known tree to obtain the maximum likelihood.
Maintaining an sDAG with all possible edges prevents these issues.

Next we provide a proof of those properties and an example of an sDAG with missing edges without those properties.
Recall by ``compatible subsplits'', we mean two subsplits $t$ and $s$ such that $s$ bipartitions one of the subsplit-clades of $t$, and an sDAG is missing an edge if it contains compatible subsplits $t$ and $s$ but not the edge $t\to s$.

\begin{prop*}
Suppose $\sdag$ is an sDAG with edges between all compatible subsplits, $\sdag^\prime$ is an sDAG given by applying an NNI to $\sdag$, and $t^\prime\to s^\prime$ in $\sdag^\prime$ is the central edge of the NNI.
If a topology $\tau^\prime$ is in $\sdag^{\prime}$ and not in $\sdag$, then the edge $t^\prime\to s^\prime$ is in $\tau^\prime$ and there exists a topology $\tau$ in $\sdag$ such that $\tau^{\prime}$ is an NNI of $\mathcal{\tau}$.
\end{prop*}
\begin{proof}
First we recall some facts about NNIs on sDAGs that follow directly from the definition.
There are at most two subsplits in $\sdag^\prime$ that are not in $\sdag$ (if such subsplits exist, they are among $t^\prime$ and $s^\prime$).
There is at most one edge that is in $\sdag^\prime$, is not in $\sdag$, and is of the form $t^\prime\to s^*$ with $\bigcup(s^*)=\bigcup(s^\prime)$ (if such an edge exists, it is $t^\prime\to s^\prime)$.
There is at most one edge that is in $\sdag^\prime$, is not in $\sdag$, and is of the form $t^*\to s^\prime$ (if such an edge exists, it is $t^\prime\to s^\prime)$.
Additionally, since $\sdag$ has all compatible edges, we can verify a generic edge $t^*\to s^*$ is in $\sdag$ by showing only that $t^*$ and $s^*$ are in $\sdag$.
Similarly, we can verify a generic topology $\tau^*$ is in $\sdag$ by showing only that each subsplit in $\tau^*$ is in $\sdag$.
Note $\tau^\prime$ is a topology and so it also has all compatible edges.

First we prove that $t^{\prime}\to s^{\prime}$ is in $\tau^{\prime}$ by showing $t^\prime$ and $s^\prime$ are in $\tau^\prime$. 
At least one of $t^\prime$ and $s^\prime$ must be in $\tau^\prime$, as otherwise all subsplits of $\tau^\prime$ are in $\sdag$ and so $\tau^\prime$ is in $\sdag$.

Suppose $t^\prime$ is in $\tau^\prime$. 
As $t^\prime$ is not a leaf, there is some edge $t^\prime\to s^{*}$ in $\tau^\prime$ with $\bigcup(s^*) = \bigcup(s^{\prime})$. 
We have a contradiction if $s^*\not=s^\prime$, as $s^*\not=s^\prime$ implies $t^\prime\to s^*$ is in $\sdag$, meaning all subsplits of $\tau^\prime$ are in $\sdag$ and so $\tau^\prime$ is in $\sdag$.
Thus $s^*=s^\prime$, so $s^\prime$ is in $\tau^\prime$.

Suppose $s^\prime$ is in $\tau^\prime$. 
As $s^\prime$ is not the root, there is some edge $t^{*}\to s^{\prime}$ in $\tau^\prime$.
We have a contradiction if $t^*\not=t^\prime$, as $t^*\not=t^\prime$ implies $t^*\to s^\prime$ is in $\sdag$, meaning all subsplits of $\tau^\prime$ are in $\sdag$.
Thus $t^*=t^\prime$, so $t^\prime$ is in $\tau^\prime$.

Thus both $t^\prime$ and $s^\prime$ are subsplits of $\tau^\prime$, and so $t^\prime\to s^\prime$ is in $\tau^\prime$.
To obtain $\tau$ from $\tau^\prime$, we apply the appropriate NNI to $\tau^\prime$ at $t^\prime\to s^\prime$.
Specifically, suppose the subsplits and edges in $\tau^\prime$ near $t^\prime\to s^\prime$ are $u\to t^\prime$, $t^\prime\to y$, $s^\prime\to x$, and $s^\prime\to z$.
Further suppose the NNI enlarging $\sdag$ to $\sdag^\prime$ swapped $\bigcup(y)$ with $\bigcup(z)$ at $t\to s$ (the case of $y$ with $x$ is identical).
Let $\tau$ by the topology given by $\tau^\prime$ after removing $t^\prime$ and $s^\prime$, removing $t^\prime\to s^\prime$ and the neighboring four edges, adding the subsplits $t$ and $s$, and adding the edges $u\to t$, $t\to s$, $t\to z$, $s\to x$, and $s\to y$.
All of these edges are between compatible subsplits of $\sdag$ and so $\tau$ is in $\sdag$.
Since $\tau$ is an NNI of $\tau^\prime$, $\tau^\prime$ is an NNI of $\tau$.
\end{proof}

To see how issues may arise when an sDAG is missing edges between compatible subsplits, consider the three topologies:
\begin{align*}
&(0,(((((1,(2,3)),(4,5)),(6,(7,8))),9),10)),\\
&(0,(((1,((2,6),((3,7),8))),(4,5)),(9,10))),\text{ and}\\
&(0,((((((((1,2),3)),((6,7),8)),4),5),10),9)).
\end{align*}
Let $\sdag$ be the sDAG generated by these topologies.
A series of lengthy, but trivial calculations, shows that 
\begin{itemize}
\item $\sdag$ contains only the three input topologies;
\item $\sdag$ is missing edges between compatible subsplits;
\item both the topology $(0,(((((1,2),3),((6,7),8)),(4,5)),(9,10)))$ and the topology $(0,((((1,((2,6),((3,7),8))),(4,5)),9),10))$ are in $\sdag^\prime$, the sDAG obtained by enlarging $\sdag$ with the NNI swapping $\{4,5\}$ and $\{6,7,8\}$ at the edge 	$\pcsp{\subsplit{1,2,3,4,5}{6,7,8}}{\subsplit{1,2,3}{4,5}}$;
\item the topology $(0,(((((1,2),3),((6,7),8)),(4,5)),(9,10)))$ is not an NNI of any topology in $\sdag$;
\item and the topology $(0,((((1,((2,6),((3,7),8))),(4,5)),9),10))$ does not contain the central edge $\pcsp{\subsplit{1,2,3,6,7,8}{4,5}}{\subsplit{1,2,3}{6,7,8}}$.
\end{itemize}

\subsection*{Initialization of Choice Maps}

Here we formalize the concept of edge choice maps and best known tree used in the top pruning algorithm.
Suppose we have a list of trees, the phylogenetic likelihoods of these trees, and the sDAG constructed from the trees. 
We define a rootward choice map as a map from each edge of the sDAG to the sibling and parent edges taken from the maximum likelihood input tree containing the edge.
Similarly, we define a leafward choice map as a map from each edge of the sDAG to the two child edges taken from the maximum likelihood input tree containing the edge.
Note these choices are made at the level of sDAG edges, not subsplits.
We also store branch lengths for every sDAG edge, where an edge of the sDAG takes its branch lengths from the maximum likelihood input tree that contains the edge.

Given these maps we apply them recursively given a starting edge, filling out a topology.
Since the edges also have assigned branch lengths, we have a tree that is ready for likelihood evaluation.
For a given edge, the tree constructed from the choice maps and branch lengths is what we take as the \emph{best known tree} containing the edge.
With multiple input trees, the best known tree for an edge produced by the choice maps need not be one of the input trees, as seen in Figure \ref{fig:choiceMapInitialization}.

We use the so-called best known tree for an edge instead of the maximum likelihood tree containing the edge because the latter is not computationally feasible.
Consider an edge from a parent subsplit $t$ to a child subsplit $s$. 
Intuitively, we can construct all topologies in the sDAG that contain the edge $t\rightarrow s$ by choosing neighboring edges and moving outward until we have constructed a full topology.
We would explore all combinations of an edge ending in the parent $t$, an edge leaving the parent $t$ to a subsplit bipartitioning the clade $\bigcup(t)-\bigcup(s)$, an edge leaving the child $s$ to a subsplit bipartitioning one subsplit-clade of $s$, and an edge leaving the child $s$ to a subsplit bipartitioning the other subsplit-clade.
We then continue outward exploring all options of neighbors for the new edges. 
This is depicted in Figure~\ref{fig:sdagNeighboringEdges}.
Furthermore, we would need to calculate the likelihoods of all of these trees.
Our best known trees greatly restrict how often we choose neighboring edges. 
In particular, for each edge of the sDAG we make the choice of four neighbors only once. 

\begin{figure}[!h]
\centering
\includegraphics[width=0.9\textwidth]{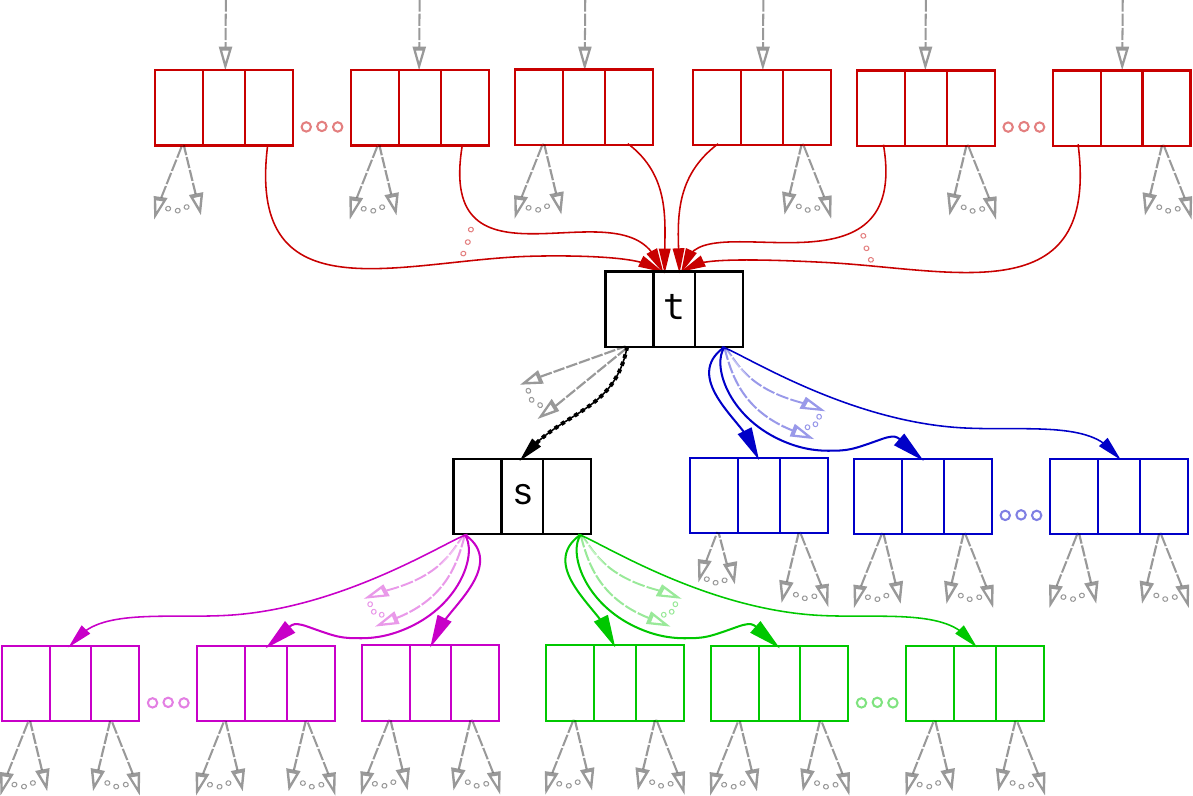}
\caption{The sDAG structure near the edge $t\rightarrow s$, which is selected to start.
The parent subsplits of $t$ in the sDAG are in red; the child subsplits of $t$ opposite $s$ are in blue; the child subsplits of $s$ bipartitioning one subsplit clade of $s$ are in green; and the child subsplits of $s$ bipartitioning the other subsplit clade of $s$ are in purple.
Constructing all topologies with the edge $t\rightarrow s$ begins by taking all combinations of selecting one of each of the 
red, blue, green, and purple subsplits and edges.}
\label{fig:sdagNeighboringEdges}
\end{figure}

At this step, we have defined choice maps and best known trees for edges in an sDAG generated from a specific list of trees.
In the next subsection, we explain how we maintain these choice maps as we grow an sDAG with NNIs.

\subsection*{Maintaining Choice Maps}

Suppose we are in the situation where we have an sDAG $\mathcal{D}$ with fully defined choice maps and branch lengths.
For an NNI on $\mathcal{D}$ producing $\mathcal{D}^\prime$, we must define the choice maps on the edges of $\mathcal{D}^\prime$. 
For edges common to $\mathcal{D}$ and $\mathcal{D}^\prime$, we use the choice maps of $\mathcal{D}$.
The remaining edges are of the five types discussed in the subsection Performing NNIs to the subsplit DAG.
Suppose the sDAGs and edges are as depicted in Figure \ref{fig:topPruningChoices}. 
Let $u$, $x$, $y$, and $z$ denote the subsplits of $\mathcal{D}$ where $u\to t$, $s\to x$, $s\to y$, and $t\to z$ are the parent, left child, right child, and sibling edges given by the choice maps at $t\to s$.
Let $u^*$, $x^*$, $y^*$, and $z^*$ denote arbitrary subsplits with $u^*\to t$, $s\to x^*$, $s\to y^*$, and $t\to z^*\in\mathcal{D}$.

We define choice maps for edges with these subsplits, $t^\prime$, and $s^\prime$
(where $t^\prime \to s^\prime$ is the central edge of the NNI)
as follows:
\begin{align}\label{eq:topPruningChoices}
&\text{parent} (u^*\rightarrow t^\prime)  = \text{parent}(u^* \rightarrow t),&
&\text{sibling} (u^*\rightarrow t^\prime) = \text{sibling}(u^* \rightarrow t),\\\nonumber
&\text{child}_1 (u^*\rightarrow t^\prime) = t^\prime \rightarrow s^\prime,&
&\text{child}_2 (u^*\rightarrow t^\prime) = t^\prime \rightarrow y,\\\nonumber
&\text{branch\_length} (u^*\rightarrow t^\prime) = \text{branch\_length}(u^* \rightarrow t),\hspace{-10em}&&\\\nonumber
&\text{parent} (t^\prime\rightarrow s^\prime) = u \rightarrow t^\prime,& 
&\text{sibling} (t^\prime\rightarrow s^\prime) = t^\prime \rightarrow y,\\\nonumber 
&\text{child}_1 (t^\prime\rightarrow s^\prime) = s^\prime \rightarrow x,& 
&\text{child}_2 (t^\prime\rightarrow s^\prime) = s^\prime \rightarrow z,\\\nonumber 
&\text{branch\_length} (t^\prime\rightarrow s^\prime) = \text{branch\_length}(t \rightarrow s),\hspace{-10em}&&\\\nonumber
&\text{parent} (t^\prime\rightarrow y^*) = u \rightarrow t^\prime,&
&\text{sibling} (t^\prime\rightarrow y^*) = t^\prime \rightarrow s^\prime,\\\nonumber
&\text{child}_1 (t^\prime\rightarrow y^*) = \text{child}_1(s \rightarrow y^*),&
&\text{child}_2 (t^\prime\rightarrow y^*) = \text{child}_2(s \rightarrow y^*),\\\nonumber
&\text{branch\_length} (t^\prime\rightarrow y^*) = \text{branch\_length}(s \rightarrow y^*),\hspace{-10em}&&\\\nonumber
&\text{parent} (s^\prime\rightarrow x^*) = t^\prime \rightarrow s^\prime,&
&\text{sibling} (s^\prime\rightarrow x^*) = s^\prime \rightarrow z,\\\nonumber
&\text{child}_1 (s^\prime\rightarrow x^*) = \text{child}_1(s \rightarrow x^*),&
&\text{child}_2 (s^\prime\rightarrow x^*) = \text{child}_2(s \rightarrow x^*),\\\nonumber
&\text{branch\_length} (s^\prime\rightarrow x^*) = \text{branch\_length}(s \rightarrow x^*),\hspace{-10em}&&\\\nonumber
&\text{parent} (s^\prime\rightarrow z^*) = t^\prime \rightarrow s^\prime,&
&\text{sibling} (s^\prime\rightarrow z^*) = s^\prime \rightarrow x,\\\nonumber 
&\text{child}_1 (s^\prime\rightarrow z^*) = \text{child}_1(t \rightarrow z^*),&
&\text{child}_2 (s^\prime\rightarrow z^*) = \text{child}_2(t \rightarrow z^*),\\\nonumber
&\text{branch\_length} (s^\prime\rightarrow z^*) = \text{branch\_length}(t\rightarrow z^*).\hspace{-10em}&&
\end{align}
We emphasize that for edges common to $\mathcal{D}$ and $\mathcal{D}^\prime$, we use the values from $\mathcal{D}$, not those in equation \eqref{eq:topPruningChoices}.

In summary, given a pre-NNI sDAG with defined choice maps, we extend the choice maps to the post-NNI sDAG.
For the branch lengths, we use those in equation \eqref{eq:topPruningChoices} as starting values for optimization. 
We optimize these branch lengths to maximize the standard phylogentic likelihood of the best known tree associated to each edge. 
For edges present in $\mathcal{D}$, the branch lengths are unaltered.
The edge $t^\prime\rightarrow s^\prime$ of $\mathcal{D}^\prime$ has a well-defined best known tree.
We define the top pruning likelihood of the NNI to be the likelihood of this tree. 

\begin{figure}[!th]\centering
\includegraphics[scale=0.5]{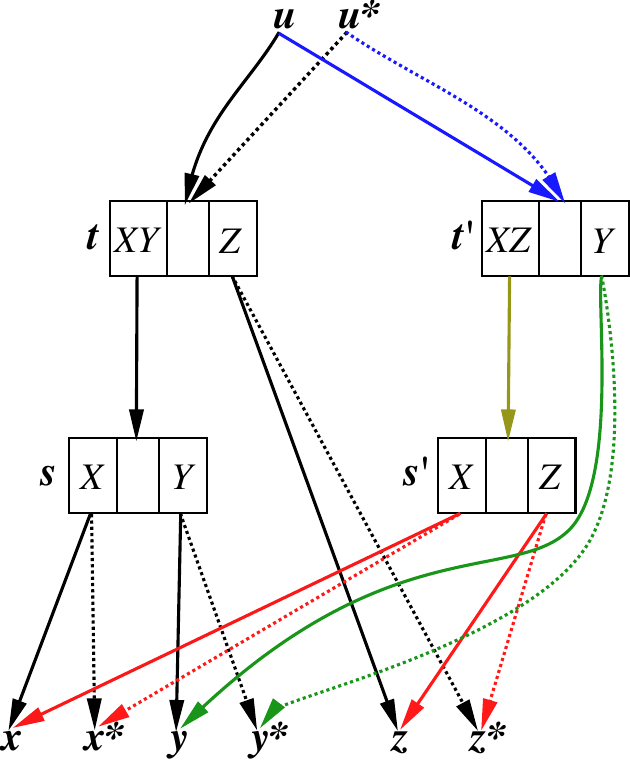}
\caption{
Subsplits and edges for the NNI swapping clades $Y$ and $Z$ at the edge $t\to s$.
The potentially new subsplits are $t^\prime$ and $s^\prime$.
Edges in black are existing edges of the sDAG, with solid lines indicating edges selected by the choice maps at $t\to s$ and dotted lines are additional edges.
Specifically, the existing choice maps take $\text{parent}(t\rightarrow s) = u\rightarrow t$, $\text{sibling}(t\rightarrow s) = t\rightarrow z$, $\text{child}_1(t\rightarrow s) = s\rightarrow x$, and $\text{child}_2(t\rightarrow s) = s\rightarrow y$.
Edges in color are potentially new edges, with solid lines being edges selected by choice maps at new edges and dotted lines
are additional edges.
}
\label{fig:topPruningChoices}
\end{figure}

To maintain an sDAG with edges between all compatible subsplits, we require choice maps and branch lengths at two additional types of edges.
These edges take the either the form $t^\prime\to s^*$, with $\bigcup(s^*)=\bigcup(s^\prime)$, or $t^*\to s^\prime$.
Denote the left subsplit-clade of $s^*$ by $X^*$, the right subsplit-clade of $s^*$ by $Z^*$, and the subsplit-clade opposite $s^\prime$ of $t^*$ by $Y^*$.
The choice maps for such edges are given by,
\begin{align}\label{eq:tpExtraEdgeChoices}
\text{parent}(t^\prime\to s^*) &= \text{parent}(t^\prime\to s^\prime),\nonumber\\
\text{sibling}(t^\prime\to s^*) &= \text{sibling}(t^\prime\to s^\prime),\nonumber\\
\text{child}_1(t^\prime\to s^*) &= \argmax_{\substack{\, s^*\to x^*\in\sdag,\\\,\bigcup(x^*)=X^*}}\, p_\psi(\bY\mid \mathcal{B}(s^*\to x^*)),\nonumber\\
\text{child}_2(t^\prime\to s^*) &= \argmax_{\substack{\, s^*\to z^*\in\sdag,\\\, \bigcup(z^*)=Z^*}}\, p_\psi(\bY\mid \mathcal{B}(s^*\to z^*)),\nonumber\\
\text{parent}(t^*\to s^\prime) &= \argmax_{u^*\to t^*\in\sdag}\, p_\psi(\bY\mid \mathcal{B}(u^*\to t^*)),\nonumber\\
\text{sibling}(t^*\to s^\prime) &= \argmax_{\substack{\,t^*\to y^*\in\sdag,\\\, \bigcup(y^*)=Y^*}}\, p_\psi(\bY\mid \mathcal{B}(t^*\to y^*)),\nonumber\\
\text{child}_1(t^*\to s^\prime) &= \text{child}_1(t^\prime\to s^\prime),\nonumber\\
\text{child}_2(t^*\to s^\prime) &= \text{child}_2(t^\prime\to s^\prime),
\end{align}
where $\mathcal{B}(e)$ denotes the best known tree for the edge $e$.
With the choice maps defined, we then assign branch lengths to $t^\prime\to s^*$ and $t^*\to s^\prime$ to maximize the phylogenetic likelihood of the best known trees for these edges.
This allows us to extend choice maps and branch lengths of $\sdag$ to those for an sDAG obtained by an NNI and adding all compatible edges.

\subsection*{How the sDAG Captures the Posterior}

Thinking in a more general sense, we can ask the question, ``does the sDAG help us find additional trees in the topological posterior distribution?'' 
With top pruning, the answer is yes, but what about building the sDAG from a collection of trees sampled from the posterior?
(Note in this section, in contrast to the main body of the paper, we do not add all compatible edges.)
Necessarily the posterior density of topologies in an sDAG is at least that of the topologies used to construct the sDAG. 
We also know that the number of topologies in an sDAG grows very fast with the number of input topologies.
But how much is the additional posterior density and how many additional credible topologies are in the sDAG? 

Starting from the beginning of an MCMC run on the data set, let $S$ be the set of topologies sampled after $k$ MCMC generations, $\mathcal{D}$ be the sDAG built from $S$, and $T$ be the set of topologies in $\mathcal{D}$.
How large are: the posterior density of trees in $S$, the posterior density of trees in $T$, and $S \cap C$ relative to $T \cap C$, where $C$ is the $95\%$ credible set?
Rather than plotting these values against $k$, the number of generations, we use the number of distinct topologies found after $k$ generations.

\begin{figure}[!t]\centering
\includegraphics[scale=0.35]{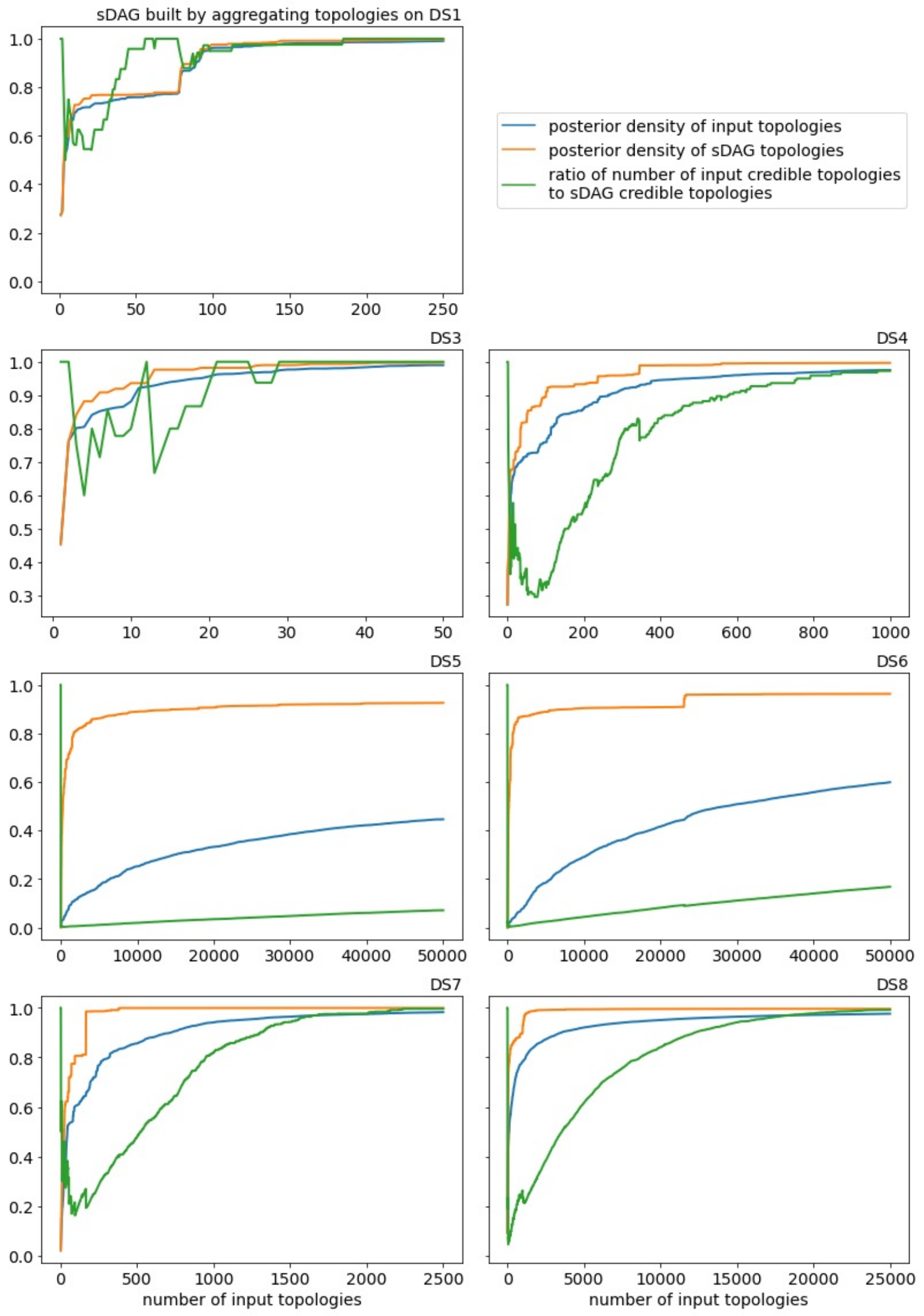}
\vspace{-1ex}
\caption{Topologies and spanned sDAGs from MCMC exploration of the DS-datasets. 
Input topologies are from the short runs of \mrbayes\ described earlier.
}
\label{fig:mcmc_to_sdag_stats}
\end{figure}

Let us examine these plots (Figure~\ref{fig:mcmc_to_sdag_stats}) in terms of posterior density, $S$, $T$, and $C$.
For posterior density, the gain in using an sDAG built from the input trees is seen by comparing the blue lines to the orange lines.
These gains vary by data set, with the more diffuse posteriors yielding higher gains.

The green lines, which relate to the credible set $C$, are harder to interpret. All plots begin in the top-left corner, because the short MCMC runs start at the maximum posterior tree (i.e., the first topology is credible). The behavior past that is again dependent on how diffuse is the posterior distribution. For DS1 and DS3 (the least diffuse), the sDAG initially provides additional credible topologies, but these topologies are quickly found by the short MCMC run. For DS5 and DS6 (the most diffuse), the sDAG provides a large number of credible topologies not found by the short MCMC run. 
On the remaining data sets, which are DS4, DS7, and DS8, the extent to which the sDAG contains additional credible topologies follows the ranking of how diffuse these data sets are.

Overall, the pattern is that for a diffuse data set, taking reasonable topologies and forming an sDAG could be very advantageous. However, if the posterior is rather compact, then the sDAG gives some additional information but at costly price in terms of the number of topologies.

\end{document}